\expandafter\def\csname ver@fixltx2e.sty\endcsname{}
\documentclass[journal,10pt]{IEEEtran}

\usepackage{amsthm}
\usepackage{amsfonts}
\usepackage{amssymb,bm}
\usepackage{mathrsfs}
\usepackage{mathtools}
\usepackage{amsmath}
\usepackage{subcaption}
\usepackage{cite}
\usepackage{balance}
\usepackage{url}
\usepackage{graphicx}
\usepackage{textcomp}
\usepackage{accents}
\usepackage{xcolor}
\newlength\mylength
\usepackage{array}
\newcolumntype{P}[1]{>{\centering\arraybackslash}p{#1}}

\usepackage{xcolor}
\usepackage{tabularx}

\newtheorem{theorem}{Theorem}
\newtheorem{lemma}{Lemma}
\newtheorem{proposition}[theorem]{Proposition}

\newtheorem{remark}{Remark}

\def\BibTeX{{\rm B\kern-.05em{\sc i\kern-.025em b}\kern-.08em
    T\kern-.1667em\lower.7ex\hbox{E}\kern-.125emX}}

\begin{document}

\title{Multi-target Range, Doppler and Angle estimation in MIMO-FMCW Radar with Limited Measurements}
\author{Chandrashekhar Rai$^\ast$, Himali Singh$^\ast$ and Arpan Chattopadhyay
\thanks{$^\ast$Co-first authors.}
\thanks{All authors are with the Electrical Engineering Department, Indian Institute of Technology (IIT) Delhi, India.  Email: \emph{csrai.cstaff@iitd.ac.in,eez208426@ee.iitd.ac.in,arpanc@ee.iitd.ac.in}.}
\thanks{The conference precursor of this work is \cite{singh2023multi}.}
\thanks{The work was supported via grant no. CRG/2022/003707 from Science and Engineering Research Board (SERB), India, project no. RP04860N from SYSTRA MVA Consulting India Private Ltd.
India, grant no. IFC/7150/2023 from Indo-French Centre for the Promotion of Advanced  Research, grant no. GP/2021/ISSC/022 from I-Hub Foundation for Cobotics, India, and Project No. FT/2024/11/37 from QUALCOMM TECHNOLOGIES, INC, USA. Himali Singh was additionally supported via Prime Minister Research Fellowship.}
}

\maketitle

\begin{abstract}
Multiple-input multiple-output (MIMO) radar offers several performance and flexibility advantages over traditional radar arrays. However, high angular and Doppler resolutions necessitate a large number of antenna elements and the transmission of numerous chirps, leading to increased hardware and computational complexity. While compressive sensing (CS) has recently been applied to pulsed-waveform radars with sparse measurements, its application to frequency-modulated continuous wave (FMCW) radar for target detection remains largely unexplored. In this paper, we propose a novel CS-based multi-target localization algorithm in the range, Doppler, and angular domains for MIMO-FMCW radar, where we jointly estimate targets’ velocities and angles of arrival. To this end, we present a signal model for sparse-random and uniform linear arrays based on three-dimensional spectral estimation. For range estimation, we propose a discrete Fourier transform (DFT)-based focusing and orthogonal matching pursuit (OMP)-based techniques, each with distinct advantages, while two-dimensional CS is used for joint Doppler-angle estimation. Leveraging the properties of structured random matrices, we establish theoretical uniform and non-uniform recovery  guarantees with high probability for the proposed framework. Our numerical experiments demonstrate that our methods achieve similar detection performance and higher resolution compared to conventional DFT and MUSIC with fewer transmitted chirps and antenna elements.
\end{abstract}

\begin{IEEEkeywords}
FMCW radar, MIMO radar, range-Doppler-angle estimation, sparse linear arrays, 2D-compressive sensing.
\end{IEEEkeywords}

\section{Introduction}\label{sec:intro}
Sensing the environment using radar or other sensors is an integral part of many engineering applications. A radar usually localizes a target of interest by estimating its range, velocity, and angle of arrival (AoA). Owing to their portability, low cost, and high range resolution, frequency-modulated continuous wave (FMCW) radars are often preferred in short-range applications like advanced driving assistance systems (ADAS)\cite{sun2020mimo,patole2017automotive}, synthetic aperture radars (SARs)\cite{meta2007signal,gu2017new}, surveillance systems\cite{saponara2017radar,kim2019low}, and human vital sign monitoring\cite{xu2022simultaneous}. FMCW radars transmit a finite number of linear frequency-modulated (LFM) chirps in each coherent processing interval (CPI). The receiver mixes the signal reflected from targets with the transmitted signal, producing a complex sinusoidal intermediate frequency (IF) signal. The IF signal's frequencies then determine the target ranges and velocities. To further localize targets in the angular domain, an array of multiple transmitter and receiver antennas is required. Multiple-input multiple-output (MIMO) radars transmit multiple orthogonal waveforms and jointly process the target returns across multiple receivers, providing additional degrees of freedom and higher resolution compared to conventional phased array radars, but with fewer physical antenna elements.

Traditionally, the discrete Fourier transform (DFT) has been used to estimate the frequencies present in the IF signal\cite{sun2020mimo,feger200977}. However, in this approach, the range and velocity resolutions depend on the bandwidth and CPI. Achieving high range resolution necessitates large bandwidths, while high velocity resolution requires transmitting and processing a large number of chirps. Similarly, array processing theory dictates that achieving a high angular resolution demands a large array aperture with numerous antenna elements to avoid ambiguities\cite{richards2014fundamentals}. Although MIMO technology enhances angular resolution, synthesizing a large virtual array with half-wavelength element spacing can be costly. Consequently, increasing resolution results in greater hardware and computational complexity. In order to mitigate these challenges, subspace-based parameter estimation techniques have been proposed in the literature. Multiple signal classification (MUSIC) has been introduced for range and angle estimation in \cite{manokhin2015music,belfiori20122d}. Estimation of signal parameters via rotational techniques (ESPRIT) is suggested in \cite{lemma2003analysis}, while \cite{kim2015joint} explores a joint DFT-ESPRIT framework. Array interpolation and eigenstructure methods are investigated in \cite{friedlander1996eigenstructure}.

Although subspace-based methods offer high resolutions, they have several drawbacks. These methods typically assume non-coherent sources or require additional smoothing in the coherent case. Additionally, they often necessitate prior knowledge of the number of targets to be estimated. The computational complexity of these methods escalate significantly when estimating parameters in multi-dimensional signals. Recently, compressive sensing (CS) has emerged as an efficient technique for sparse signal recovery with limited measurements\cite{elad2010sparse}. By utilizing a user-defined parameter grid, CS facilitates high-resolution estimation with a small number of measurements and low computational efforts. In this work, we utilize CS techniques in MIMO-FMCW radars for range and joint velocity-angle estimation at high resolutions but using a small number of antenna elements and transmitted chirps. Note that in radar, the sampled IF signal, chirps in each CPI, and array channels correspond to measurements across different domains, known as fast-time, slow-time, and spatial domain samples, respectively. Additionally, we investigate the theoretical recovery guarantees for our joint velocity-angle estimation framework and demonstrate the performance of our proposed method through extensive numerical experiments.

\subsection{Prior Arts}\label{subsec:prior art}
Earlier CS applications in radar systems primarily aimed to enhance resolution while still relying on uniform sampling or full measurement setups. For instance, \cite{baraniuk2007compressive} replaced the matched filter by CS using random projections to reconstruct radar images, thereby lowering ADC demands. Other works, such as \cite{herman2008compressed,herman2009high,chen1989orthogonal}, explored pseudo-random phase-coded transmit waveforms to improve resolution in delay-Doppler and range-angle domains. In \cite{potter2010sparsity}, random frequency-hopping was employed for transmit waveform design in MIMO radar and SAR, while \cite{yoon2008compressed} utilized randomized waveforms for super-resolution imaging. These approaches also utilized randomly selected measurements across time, frequency, or spatial domains. In \cite{anitori2012design}, a CS-based energy detector was applied to the recovered outputs, while \cite{bilik2011spatial} exploited spatial diversity through varying array orientations for high-resolution spectrum estimation, with randomness introduced by the dynamic sensor positioning. More recently, \cite{tohidi2020compressed} proposed a GLRT-based detection method operating directly on compressed measurements acquired through random projections and selective sampling. However, these works acquire measurements at high uniform sampling rates, similar to conventional radars, and process only the selected or linearly projected measurements. While this can reduce computational load and improve resolution, it does not significantly reduce hardware complexity. In contrast, we adopt a fundamentally different approach by using a reduced number of physical antenna elements randomly distributed over the array aperture and randomly transmitting only a limited set of chirps in each CPI. Hence, we achieve high-resolution performance while substantially lowering both computational and hardware requirements. Note that our random chirp setup also differs from the difference co-chirps proposed in \cite{xu2023automotive}, where the sparse chirp set arises from difference co-array structures, such as co-prime and nested chirps, instead of being transmitted randomly.

A key benefit of CS techniques is the reduction in measurement requirements for sparse signals. In radar, the target scene is sparse because only a few targets are typically present, making CS a natural fit for the problem. As a result, CS has been widely adopted to enable sub-Nyquist radars with reduced measurements \cite{bar2014sub,cohen2018summer,mishra2019sub}. In the spatial domain, sparse linear arrays (SLAs), with fewer antenna elements compared to uniform linear arrays (ULAs), have been introduced for both pulsed and continuous-wave radars \cite{feger200977,sun2020sparse,diamantaras2021sparse}. Optimal sparse array design was investigated in \cite{diamantaras2021sparse}, while \cite{feger200977} designed a non-uniform SLA and applied digital beamforming for AOA estimation after interpolating the missing measurements. Conversely, \cite{sun2020sparse} suggested matrix completion methods to reconstruct the corresponding linear array. Unlike the random selection or projection-based approaches discussed earlier, these matrix completion and interpolation methods directly acquire sparse measurements. However, these techniques still complete the missing measurements before estimating target parameters, leading to increased computational complexity at the receiver. Contrarily, our approach estimates target parameters directly from the support of the recovered sparse signal, without requiring complete measurements. This enables not only reduced computational complexity and hardware requirements but also achieves superior resolution compared to conventional methods.

Alternatively, spatial CS enables direct parameter recovery from SLAs\cite{rossi2013spatial,yu2010mimo}. For instance, \cite{rossi2013spatial} focused on pulsed-MIMO radar and estimated AOAs for a specific range-velocity bin using different CS recovery algorithms. Additionally, \cite{liu2015three} incorporated velocity estimation but not range. In \cite{yu2010mimo}, velocities and angles were estimated using measurements from a small number of randomly positioned transmitters and receivers on a circular disc. Besides spatial compression, CS techniques have also been applied in radar systems for interference mitigation \cite{correas2019sparse}, spectrum sharing \cite{cohen2017spectrum}, joint radar-communication systems\cite{ma2021frac}, and multi-target shadowing effect mitigation\cite{cao2021compressed}. Furthermore, \cite{ender2010compressive} discussed the advantages and challenges of applying CS in radar, including clutter cancellation. However, these earlier studies focused primarily on pulsed-wave radars and were limited to estimating at most two target parameters. Conversely, our work develops CS-based techniques for estimating range, velocity, and AOA in MIMO-FMCW radars. To this end, we introduce a joint 2D-CS framework for Doppler-angle recovery, which has not been previously explored in the literature.

\subsection{Our contributions} \label{subsec:contributions}
Preliminary results of this work appeared in our conference publication \cite{singh2023multi}, where we addressed range and angle estimation only, without any theoretical guarantees. In this work, we consider a MIMO-FMCW radar with a random SLA that transmits only a subset of chirps per CPI and present a joint Doppler-angle estimation framework. The reduced number of antenna elements and chirps correspond to spatial and slow-time CS, respectively. To the best of our knowledge, the use of randomly transmitted sparse chirps, their integration with random SLAs for joint Doppler–angle estimation in MIMO-FMCW radars, and the associated recovery guarantees have not been investigated previously. These aspects constitute the key novel contributions of our work, as follows:\\
\textbf{1) Range estimation:} Prior works on CS-based radars have mainly focused on pulsed-wave radars and/or recovering target angles and velocities within a specific range bin. In this paper, we first present a separable mixture model for the IF signal, applicable to both full and sparse measurements. We propose two methods for range estimation: (a) DFT-focusing followed by binary integration\cite{richards2014fundamentals}, and (b) range-orthogonal matching pursuit (Range-OMP). In contrast to the coherently-integrated chirps in the conventional approaches, binary integration enhances detection performance at low SNRs, while the focusing operation\cite{bar2014sub} concentrates all the target returns from the same range in a single DFT-bin irrespective of their velocities and AOAs. However, in both conventional and proposed DFT-based range estimation, resolution is tied to the DFT-defined range bins such that achieving finer resolution requires increased bandwidth and results in higher computational costs. To this end, we propose Range-OMP, which uses a user-defined grid and a greedy OMP algorithm to deliver higher resolution with lower computational complexity while relying on measurements from just a single chirp and array channel; see further Remarks~\ref{remark:coherent binary}-\ref{remark:complexity_range}. Unlike standard OMP, which reconstructs sparse signals from compressed measurements, our Range-OMP leverages full fast-time radar measurements obtained through uniform sampling of the IF signal.\\
\textbf{2) Joint Doppler-angle estimation:} We jointly estimate target velocities and AOAs within each detected range bin using both vectorized and joint 2D-CS techniques. In conventional radar systems, achieving finer velocity and angular resolution typically demands transmitting a large number of chirps per CPI and deploying wide aperture arrays with many physical antenna elements—both of which significantly increase computational burden. Although super-resolution methods like MUSIC can enhance resolution without increasing measurements, they generally incur higher computational costs compared to simpler DFT-based approaches. To overcome these challenges, we adopt sparse random chirps and random SLAs, thereby reducing both slow-time and spatial measurements, and develop a joint CS-based velocity-angle estimation framework. Despite this reduction, our methods achieve performance on par with full-measurement systems, aided by user-defined grids that enable high-resolution estimation of velocity and AOA. Furthermore, the reduced measurements lower computational complexity and allow the radar to scan multiple angular sectors within a single CPI. We further summarize these advantages over conventional methods in Table~\ref{tbl:comparison}.\\
\textbf{3) Recovery guarantees:} CS provides approximate solutions for the sparse recovery problem, with sufficient conditions for high-probability recovery being widely studied. In our work, we examine both uniform and non-uniform guarantees and establish bounds on the number of measurements and recovery errors. Specifically, we show that our random but structured measurement matrix exhibits low coherence and satisfies the isotropy property, provided the transmitted chirps, random antenna elements, and velocity-angle grids meet suitable conditions. Deriving these theoretical guarantees is challenging due to the increased dimensionality of the recovered parameters and the interdependence of rows and columns in the measurement matrix. We also present a practical MIMO-FMCW radar setup that satisfies these conditions.\\
\textbf{4) Comprehensive evaluation:} We consider the detection performance of our CS-based methods, including receiver operator characteristic (ROC) and time complexity, and compare them to the classical-DFT and subspace-based MUSIC methods, which rely on full measurements. Our approach achieves similar detection probabilities as these full-measurement techniques but with only half the number of transmitted chirps and physical antenna elements. Additionally, Range-OMP and CS-based joint Doppler-angle estimation outperform DFT-based methods in terms of accuracy, while offering resolution comparable to MUSIC at a significantly lower computational complexity.

The rest of the paper is organized as follows. The next section introduces the MIMO-FMCW radar's system model generalized for both full and sparse measurements. Section~\ref{sec:methodology} develops the CS-based parameter estimation methods while Section~\ref{sec:guarantees} provides the theoretical recovery guarantees. In Section~\ref{sec:numericals}, we illustrate the performance of the proposed methods through extensive numerical experiments before concluding in Section~\ref{sec:conclusions}.

Throughout the paper, we reserve boldface lowercase and uppercase letters for vectors (column vectors) and matrices, respectively, and $\lbrace a_{i}\rbrace_{i_{1}\leq i\leq i_{2}}$ (or simply $a_{i_{1}\leq i\leq i_{2}}$) denotes a set of elements indexed by an integer $i$. The notations $[\mathbf{A}]_{:,i}$, $[\mathbf{A}]_{i,:}$ and $[\mathbf{A}]_{i,j}$ denote the $i$-th column, $i$-th row and $(i,j)$-th element of matrix $\mathbf{A}$, respectively. The transpose/ Hermitian/conjugate operation is $(\cdot)^{T/H/*}$, expectation/probability is $\mathbb{E}/\mathbb{P}[\cdot]$, the outer product is $\otimes$ and the complement of a set is $(\cdot)^{c}$. The $l_{2/1/0}$ norm of a vector is $\|\cdot\|_{2/1/0}$. The notations $\textrm{supp}$, $\textrm{diag}$, and $\textrm{vec}$ denote support, diagonal matrix, and vectorization, respectively. Also, $\mathbf{I}_{n}$ and $\mathbf{0}$ denote a `$n\times n$' identity matrix and an all-zero matrix, respectively. We represent the circular-normal and real-valued normal distributions as $\mathcal{CN}(\bm{\mu},\mathbf{Q})$ and $\mathcal{N}(\bm{\mu},\mathbf{Q})$ (with mean $\bm{\mu}$ and covariance matrix $\mathbf{Q}$) while $\mathcal{U}[a,b]/\mathcal{U}\{a_{i}\}_{i_{1}\leq i\leq i_{2}}$ represents a uniform/ discrete-uniform distribution over interval $[a,b]$/ set $\{a_{i}\}_{i_{1}\leq i\leq i_{2}}$.

\section{System Model}\label{sec:system model}
Consider a monostatic MIMO radar system, as shown in Fig.~\ref{fig:radar setup}a, consisting of $N_{T}$ transmitters and $N_{R}$ receivers located over a (possibly overlapping) array of apertures $A_{T}$ and $A_{R}$, respectively. We define $A\doteq A_{T}+A_{R}$. The $n$-th transmitter and $m$-th receiver are located at $A\alpha_{n}/2$ and $A\beta_{m}/2$ (along the array length), respectively, where $\alpha_{n}\in[-A_{T}/A,A_{T}/A]$ and $\beta_{m}\in[-A_{R}/A,A_{R}/A]$. For a random SLA, $\{\alpha_{n}\}_{1\leq n\leq N_{T}}$ and $\{\beta_{m}\}_{1\leq m\leq N_{R}}$ are drawn i.i.d. from distributions $\mathcal{P}_{\alpha}$ and $\mathcal{P}_{\beta}$, respectively. In a ULA, $\{\alpha_{n}\}$ and $\{\beta_{m}\}$ correspond to the uniformly-spaced transmitter and receiver locations, which results in a virtual array of half-wavelength element spacing, i.e., the spatial Nyquist sampling rate. Note that our random SLA differs from the index-modulation-based joint radar-communication setup of \cite{ma2021frac}, wherein $\{\alpha_{n}\}$ are selected according to the message communicated while $\{\beta_{m}\}$ correspond to a uniform receiver array.
\begin{figure*}
  \centering
  \includegraphics[width = 0.75\linewidth]{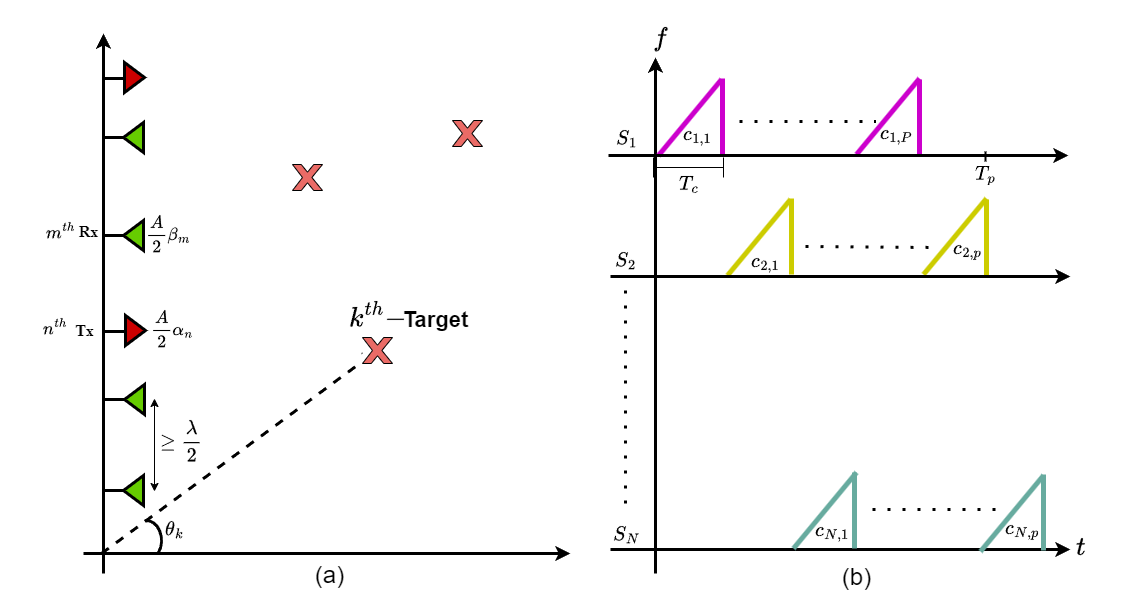}
  \caption{MIMO-FMCW radar setup with sparse chirps and antenna array: (a) random SLA (green and red triangles denote receivers and transmitters, respectively); and (b) Time-frequency illustration of sparse random chirps transmitted to different angular sectors ($S_{j}$ and $c_{j,p}$ denote $j$-th sector and $p$-th chirp for the $j$-th sector, respectively).}
 \label{fig:radar setup}
\end{figure*}

The transmitters transmit orthogonal LFM chirps of carrier frequency $f_{c}$, chirp rate $\gamma$, and chirp duration $T_{c}$. The wavelength is $\lambda=c/f_{c}$, where $c$ represents the speed of light. For simplicity, we consider time-domain orthogonality, i.e., the transmitters transmit the same signal with relative time shifts. Hence, the received signal components corresponding to various transmitters are readily separated at different receivers. Alternatively, orthogonal waveforms for MIMO-FMCW radars have been designed using beat frequency and chirp rate offsets \cite{de2011orthogonal,babur2013nearly}. In \cite{xu2023automotive}, difference co-chirps are proposed for high-accuracy and low complexity range-Doppler estimation in FMCW radars. We consider a CPI of duration $T_{p}=P_{max}T_{c}$ corresponding to $P_{max}$ chirps. The standard radar system with full (slow-time) measurements transmits for the entire CPI. In this work, the radar transmits only $P$ out of $P_{max}$ chirps ($P<P_{max}$) in a specific angular sector of interest. We denote the sparse chirp set by $\mathcal{P}_{s}=\{\zeta_{p}\}_{1\leq p\leq P}$ where $\zeta_{p}$ are randomly drawn distinct integers from $\{0,1,\hdots,P_{max}-1\}$ and distributed as $\mathcal{P}_{p}$. On the other hand, in the standard radar, the set $\{\zeta_{p}\}=\{0,1,\hdots, P_{max}-1\}$ and all chirps are transmitted in the same angular sector. In our sparse radar setup, the time corresponding to $\mathcal{P}_{s}^{c}$, i.e., the non-transmitted chirps, can be utilized in other angular sectors, as depicted in Fig.~\ref{fig:radar setup}b. Throughout the paper, we consider the target's parameter estimation in a single sector only. The received signals from other sectors can be trivially processed in the same manner.

The FMCW radar's LFM chirp transmitted by each transmitter is modeled as
\par\noindent\small
\begin{align*}
    s(t)=\exp{\left(j2\pi\left(f_{c}t+\frac{\gamma}{2}t^{2}\right)\right)},\;\; 0\leq t\leq T_{c}.
\end{align*}
\normalsize
We consider a target scene of $K$ far-field, non-fluctuating point targets with the $k$-th target's range, radial velocity, and AOA denoted by $R_{k}$, $\nu_{k}$, and $\theta_{k}$, respectively. In our proposed framework, the received signal from each transmitted chirp is processed independently at every receiver. Hence, we first focus on the received signal component at the $m$-th receiver corresponding to the $p$-th chirp transmitted from $n$-th transmitter, given by
\par\noindent\small
\begin{align*}
    r_{n,m,p}(t)=\sum_{k=1}^{K}a_{k} s(t-\tau^{k}_{n,m,p}),\;\; \zeta_{p}T_{c}\leq t< (\zeta_{p}+1)T_{c},
\end{align*}
\normalsize
where $a_{k}$ is the complex amplitude proportional to the $k$-th target's radar cross-section (RCS) and $\tau^{k}_{n,m,p}$ is the total delay in the $k$-th target's return. The delay $\tau^{k}_{n,m,p}$ consists of range delay $\tau^{R}_{k}$, Doppler delay $\tau^{D}_{p,k}$ and angular delay $\tau^{\theta}_{n,m,k}$ as
\par\noindent\small
\begin{align}
    \tau^{k}_{n,m,p}=\tau^{R}_{k}+\tau^{D}_{p,k}+\tau^{\theta}_{n,m,k},\label{eqn:delay components}
\end{align}
\normalsize
where $\tau^{R}_{k}=2R_{k}/c$, $\tau^{D}_{p,k}=2\nu_{k}\zeta_{p}T_{c}/c$ and $\tau^{\theta}_{n,m,k}=A\sin{(\theta_{k})}\times(\alpha_{n}+\beta_{m})/2c$. Note that $\tau^{D}_{p,k}$ represents the additional delay (compared to the range delay) in the target's return because of the target's receding motion in $\zeta_{p}T_{c}$ time interval at velocity $\nu_{k}$. Also, the far-field assumption results in a constant AOA across the array.

At the $m$-th receiver, the received signal $r_{n,m,p}(t)$ is mixed with the transmitted chirp $s(t)$ to obtain the IF signal as
\par\noindent\small
\begin{align*}
    &y_{n,m,p}(t)=s(t)r^{*}_{n,m,p}(t)=\sum_{k=1}^{K}a_{k}^{*}\exp{\left(j2\pi\gamma\tau^{k}_{n,m,p}t\right)}\\
    &\;\;\times\exp{\left(-j\pi\gamma(\tau^{k}_{n,m,p})^{2}\right)}\exp{\left(j2\pi f_{c}\tau^{k}_{n,m,p}\right)},
\end{align*}
\normalsize
which is then sampled at sampling frequency $f_{s}$ to yield the (discrete) fast-time measurements
\par\noindent\small
\begin{align}
    y_{n,m,p}[t]&=\sum_{k=1}^{K}a_{k}^{*}\underbrace{\exp{\left(j2\pi\gamma\tau^{k}_{n,m,p}\frac{t}{f_{s}}\right)}}_{\textrm{Term-I}}\nonumber\\
    &\times\underbrace{\exp{\left(-j\pi\gamma(\tau^{k}_{n,m,p})^{2}\right)}}_{\textrm{Term-II}}\underbrace{\exp{\left(j2\pi f_{c}\tau^{k}_{n,m,p}\right)}}_{\textrm{Term-III}}.\label{eqn:discrete IF signal}
\end{align}
\normalsize
We, henceforth, deal with only discrete-time signals and use $t$ to denote the discrete-time index with $0\leq t\leq N-1$ where $N=\lceil f_{s}T_{c}\rceil$. In a MIMO-FMCW radar with $N_{T}$ transmitters, $N_{R}$ receivers and $P$ chirps, we obtain `$N_{T}N_{R}\times P$' sampled measurements $\{y_{n,m,p}[t]\}_{1\leq n\leq N_{T}, 1\leq m\leq N_{R}, 1\leq p\leq P}$.

We now show that the target parameter estimation is, in fact, a 3D frequency estimation problem under suitable approximations for practical systems. \\
\textit{Term-I:} For practical FMCW radars with narrow-band assumption, we have $\tau^{R}_{k}\gg\tau^{D}_{p,k}$ and $\tau^{R}_{k}\gg\tau^{\theta}_{n,m,k}$ for all targets such that $\gamma\tau^{k}_{n,m,p}/f_{s}\approx\gamma\tau^{R}_{k}/f_{s}$. This can also be verified numerically for a target with the radar's parameters provided in Table \ref{tbl:parameters}. We define the normalized beat frequency (due to range) as $\Omega_{R}^{k}\doteq\gamma\tau^{R}_{k}/f_{s}$ and approximate Term-I as $\exp{(j2\pi\Omega_{R}^{k}t)}$. This term is further examined in Section~\ref{subsec:range processing} using DFT-based focusing or OMP-recovery methods to obtain the range estimates.\\
\textit{Term-II:} Again, using $\tau^{R}_{k}\gg\tau^{D}_{p,k}+\tau^{\theta}_{n,m,k}$, we have $(\tau^{k}_{n,m,p})^{2}\approx(\tau^{R}_{k})^{2}$ \cite{winkler2007range}. Hence, Term-II becomes $\exp{(-j\pi\gamma)(\tau^{R}_{k})^{2}}$, a constant for each target independent of chirp and antenna indices. In the 3D-mixture model \eqref{eqn:3D mixture}, we include this term in $\widetilde{a}_{k}$, which is not the focus of the parameter estimation problem addressed in this paper.\\
\textit{Term-III:} Substituting \eqref{eqn:delay components} in Term-III, we obtain
\par\noindent\small
\begin{align*}
    \exp{\left(j2\pi f_{c}\tau^{k}_{n,m,p}\right)}&=\exp{\left(j2\pi f_{c}\tau^{R}_{k}\right)}\exp{\left(j2\pi\frac{2\nu_{k}T_{c}}{\lambda}\zeta_{p}\right)}\\
    &\;\;\times\exp{\left(j2\pi\frac{A\sin(\theta_{k})}{2\lambda}(\alpha_{n}+\beta_{m})\right)}.
\end{align*}
\normalsize
Here, $\exp{\left(j2\pi f_{c}\tau^{R}_{k}\right)}$ does not vary with antenna elements and chirps and hence, included in $\widetilde{a}_{k}$ in \eqref{eqn:3D mixture}. Define the normalized Doppler frequency $\Omega^{k}_{D}\doteq 2\nu_{k}T_{c}/\lambda$ and normalized angular/spatial frequency $\Omega^{k}_{\theta}\doteq A\sin(\theta_{k})/2\lambda$. Term-III is investigated in Section~\ref{subsec:angle-doppler processing} to jointly estimate the target velocities and AOAs.

With these approximations and including the circular-normal noise term $w_{n,m,p}[t]$ to represent the interference and noises present in the radar measurements, \eqref{eqn:discrete IF signal} simplifies to
\par\noindent\small
\begin{align}
    y_{n,m,p}[t]&\approx\sum_{k=1}^{K}\widetilde{a}^{*}_{k}\exp{(j2\pi\Omega^{k}_{R}t)}\exp{(j2\pi\Omega^{k}_{D}\zeta_{p})}\nonumber\\
    &\;\;\times\exp{(j2\pi\Omega^{k}_{\theta}(\alpha_{n}+\beta_{m}))}+w_{n,m,p}[t],\label{eqn:3D mixture}
\end{align}
\normalsize
where $\widetilde{a}_{k}=a_{k}exp(j\pi\gamma(\tau^{R}_{k})^{2})exp(-j2\pi f_{c}\tau^{R}_{k})$. Here, \eqref{eqn:3D mixture} consists of a separable mixture of three different complex exponentials such that the target's range, velocity, and AOA estimation problem is equivalent to 3D frequency estimation.
\begin{remark}[Generalized model]\label{remark:general}
    The model \eqref{eqn:3D mixture} represents a generalized received IF signal for both full and sparse measurement radars. In particular, in the case of SLA and the sparse chirps, $\alpha_{n}$, $\beta_{m}$, and $\zeta_{p}$ are randomly drawn from suitable distributions and represent non-uniform measurements. On the other hand, in conventional radar setups, these quantities correspond to the uniformly placed antenna elements and chirps transmitted over the entire CPI.
\end{remark}

\textbf{Limited measurements:} In a typical radar system estimating range, velocity, and AOA, three distinct types of measurements are used: (a) fast-time samples, i.e., the sampled IF signal, used for range estimation, (b) slow-time samples obtained from multiple chirps within a CPI for velocity estimation, and (c) spatial samples collected across array channels for AOA estimation. In our work, we adopt the conventional fast-time sampling at frequency $f_{s}$, allowing for range estimation through DFT-based focusing method in Section~\ref{subsec:range processing}.1. Both classical-DFT and our DFT-focusing rely on DFT-defined range bins, which influence both resolution and computational cost. In contrast, Section~\ref{subsec:range processing}.2 presents the Range-OMP algorithm, which utilizes the same fast-time data to achieve finer range resolution with reduced complexity. However, unlike full-measurement systems, we limit the number of slow-time and spatial domain samples. Specifically, only $P<P_{max}$ chirps are transmitted per CPI, and our SLA, illustrated in Fig.~\ref{fig:radar setup}a, comprises of $N_{T}<N_{T,f}$ transmitters and $N_{R}<N_{R,f}$ receivers placed across apertures $A_{T}$ and $A_{R}$, respectively. Here, $N_{T,f}$ and $N_{R,f}$ denote the total transmitters and receivers in a full ULA. Traditional approaches like classical-DFT and MUSIC leverage complete measurements from all $P_{max}$ chirps and the full $N_{T,f}\times N_{R,f}$ array channels, and hence, improved resolution comes at the cost of increased measurement and processing overhead. To address this, we employ CS techniques in Section~\ref{subsec:angle-doppler processing} to jointly estimate velocity and AOA with enhanced resolution and significantly lower computational load, despite relying on fewer chirps and antenna elements. As demonstrated in Section~\ref{sec:numericals}, our proposed methods achieve superior estimation accuracy over classical-DFT and MUSIC, even when using only half the number of chirps and array elements compared to the full-measurement radar configuration.

\section{Sparse target-parameter recovery}\label{sec:methodology}
In this section, unlike prior studies, we consider both range and joint Doppler-angle estimation using CS-based methods in Sections~\ref{subsec:range processing} and \ref{subsec:angle-doppler processing}, respectively. As a precursor, Section~\ref{subsec:CS summary} provides a brief overview of the standard sparse signal recovery problem and the CS framework. Finally, Table~\ref{tbl:comparison} highlights the key distinctions between our proposed CS-based methods and the classical-DFT and MUSIC approaches.

    \begin{table*}
    \caption{Proposed CS-based methods compared with classical-DFT and MUSIC\cite{belfiori20122d}}
    \label{tbl:comparison}
    \centering
    \begin{tabular*}{\textwidth}{p{2.2cm}p{5.0cm}p{4.5cm}p{5.0cm}}
    \hline\noalign{\smallskip}
    \textbf{Detail} & \textbf{Classical-DFT} & \textbf{MUSIC} & \textbf{Proposed CS-based methods}\\
    \noalign{\smallskip}
    \hline
    \noalign{\smallskip}
    Radar setup & ULA with $\lambda/2$ element spacing ($N_{T,f}$ transmitters and $N_{R,f}$ receivers); $P_{max}$ chirps transmitted for the entire CPI  & ULA with $\lambda/2$ element spacing ($N_{T,f}$ transmitters and $N_{R,f}$ receivers); $P_{max}$ chirps transmitted for the entire CPI & Random SLA with inter-element spacing $\geq\lambda/2$ ($N_{T}<N_{T,f}$ transmitters and $N_{R}<N_{R,f}$ receivers); $P<P_{max}$ chirps transmitted randomly in a CPI\\
    Range estimation & DFT with coherent integration of all chirps and array channels & Same as classical-DFT & DFT with binary integration across all chirps and array channels, or Range-OMP using only one chirp and array channel\\
    Joint Doppler-angle estimation & 2D-DFT & 2D-MUSIC with spatial smoothing & 2D-OMP or vectorized 1D-OMP/ BP/ LASSO\\
    Range resolution & $c/2\gamma T_{c}$; increases with total bandwidth swept by LFM chirp & Same as classical-DFT & Same as classical-DFT for DFT+binary integration; higher for Range-OMP depending on the choice of range grid $\omega_{1\leq g\leq G_{R}}$\\
    Doppler \& angular resolutions & Doppler resolution ($\lambda/2P_{max}T_{c}$) increases with transmitted chirps; Angular resolution increases with array aperture & Higher than classical-DFT; depends on the array aperture, number of chirps, SNR and search grid density & Higher than classical-DFT; depends on the choice of Doppler grid $\rho_{1\leq g\leq G_{D}}$ and angular grid $\phi_{1\leq g\leq G_{\theta}}$\\
    Computational complexity & $\mathcal{O}(N_{T,f}N_{R,f}P_{max}N\log(N))$ for range estimation; $\mathcal{O}(N_{T,f}N_{R,f}P_{max}\log(N_{T,f}N_{R,f}P_{max}))$ for joint Doppler-angle estimation & Same as classical-DFT for range estimation; $\mathcal{O}(N_{T,f}^{2}N_{R,f}^{2}P_{max}+N_{T,f}^{3}N_{R,f}^{3}+G_{1}G_{2}N_{T,f}^{2}N_{R,f}^{2})$ for joint Doppler-angle estimation ($G_{1}$ \& $G_{2}$ are Doppler and angular search grid sizes, respectively) & Lower than classical-DFT and MUSIC because of smaller number of measurements ($N_{T}<N_{T,f}$, $N_{R}<N_{R,f}$, $P<P_{max}$); see Remark~\ref{remark:complexity_range} and Table~\ref{tbl:complexity}\\
    \noalign{\smallskip}
    \hline\noalign{\smallskip}
    \end{tabular*}
    \end{table*}

\subsection{CS sparse signal recovery}\label{subsec:CS summary}
Consider the linear measurement model
\par\noindent\small
\begin{align}
    \mathbf{y}=\mathbf{A}\mathbf{x}+\mathbf{w},\label{eqn:sparse 1d problem}
\end{align}
\normalsize
where $\mathbf{x}\in\mathbb{C}^{N\times 1}$ is the signal measured under the sensing matrix $\mathbf{A}\in\mathbb{C}^{M\times N}$ to yield measurements $\mathbf{y}\in\mathbb{C}^{M\times 1}$, and $\mathbf{w}\in\mathbb{C}^{M\times 1}$ is the additive noise. CS aims to recover a sparse vector $\mathbf{x}$ given a small number of measurements $\mathbf{y}$ with $M\ll N$ such that \eqref{eqn:sparse 1d problem} is an under-determined system. A vector $\mathbf{x}$ is said to be $K$-sparse if $\|\mathbf{x}\|_{0}\leq K<N$. Theoretically, a $K$-sparse $\mathbf{x}$ can be recovered from \eqref{eqn:sparse 1d problem} for $M\ll N$ by solving the non-convex combinatorial problem
\par\noindent\small
\begin{align}
    \textrm{min}_{\mathbf{x}}\;\;\|\mathbf{x}\|_{0}\;\;\textrm{s.t.}\;\;\|\mathbf{y}-\mathbf{Ax}\|_{2}\leq\epsilon,\label{eqn:CS l0 problem}
\end{align}
\normalsize
where the parameter $\epsilon$ is chosen based on the  statistics of the noise $\mathbf{w}$. However, solving \eqref{eqn:CS l0 problem} is an NP-hard problem and involves an exhaustive search of exponential complexity over all possible combinations of non-zero indices of $\mathbf{x}$ \cite{eldar2012compressed}. Hence, approximate solutions of polynomial complexity have been developed in CS.

CS algorithms can be broadly categorized into matching pursuit (MP) and basis pursuit (BP) methods. In MP or greedy methods, new indices are added recursively to the previous provisional support. OMP\cite{pati1993orthogonal}, orthogonal least squares (OLS)\cite{chen1989orthogonal} and CoSAMP\cite{needell2009cosamp} are widely used MP algorithms. On the other hand, BP relaxes the $l_{0}$-norm in \eqref{eqn:CS l0 problem} to an $l_{1}$-norm\cite{candes2008introduction} resulting in a convex minimization problem for which a globally optimal solution can be obtained in polynomial time. The standard BP and LASSO\cite{tibshirani1996regression} formulations of \eqref{eqn:CS l0 problem} are
\par\noindent\small
\begin{align}
&\textrm{BP:}\;\;\;\textrm{min}_{\mathbf{x}}\;\;\|\mathbf{x}\|_{1}\;\;\textrm{s.t.}\;\;\|\mathbf{y}-\mathbf{Ax}\|_{2}\leq\epsilon,\label{eqn:BP}\\
&\textrm{LASSO:}\;\;\;\textrm{min}_{\mathbf{x}}\;\;\|\mathbf{y}-\mathbf{Ax}\|_{2}\;\;\textrm{s.t.}\;\;\|\mathbf{x}\|_{1}\leq\kappa,\label{eqn:LASSO}
\end{align}
\normalsize
where $\kappa$ is chosen based on the desired sparsity of $\mathbf{x}$.

The measurement model \eqref{eqn:sparse 1d problem} can be generalized to 2D-separable measurements as
\par\noindent\small
\begin{align}
    \mathbf{Y}=\mathbf{A}\mathbf{X}\mathbf{B}^{T}+\mathbf{W},\label{eqn:sparse 2d problem}
\end{align}
\normalsize
where $\mathbf{X}\in\mathbb{C}^{N_{1}\times N_{2}}$ is the sparse input matrix, $\mathbf{A}\in\mathbb{C}^{M_{1}\times N_{1}}$ and $\mathbf{B}\in\mathbb{C}^{M_{2}\times N_{2}}$ are the measurement matrices, $\mathbf{Y}\in\mathbb{C}^{M_{1}\times M_{2}}$ contains  the 2D measurements with $M_{1}M_{2}\ll N_{1}N_{2}$ and $\mathbf{W}\in\mathbb{C}^{M_{1}\times M_{2}}$ is the noise term. One approach to recover $\mathbf{X}$ is to vectorize the model \eqref{eqn:sparse 2d problem} and then apply standard 1D-MP or BP methods. Alternatively, matrix $\mathbf{X}$ can be recovered directly using matrix projections as in 2D-OMP\cite{fang20122d,liu2010fast}. We consider both vectorized and 2D-OMP methods to jointly estimate the target velocities and AOAs in Section~\ref{subsec:angle-doppler processing}.

\subsection{Range estimation}\label{subsec:range processing}
We first estimate the target ranges from measurements \eqref{eqn:3D mixture} using two different methods: DFT-based focusing and OMP. Fig.~\ref{fig:range schematic} illustrates the conventional and proposed range estimation algorithms. The target velocities and AOAs are then jointly estimated for each detected range in the subsequent section.

\textbf{1) DFT-based focusing:} Consider the $N$-point normalized DFT of \eqref{eqn:3D mixture} across the fast time samples (i.e, $t$) for th $p$-th chirp and $(n,m)$-th virtual array channel as
\par\noindent\small
\begin{align}
    Y_{n,m,p}[l]&=\frac{1}{\sqrt{N}}\sum_{t=0}^{N-1}y_{n,m,p}[t]\exp{\left(-j2\pi\frac{lt}{N}\right)}\nonumber
\end{align}
\begin{align}
Y_{n,m,p}[l]&=\sum_{k=1}^{K}\widetilde{a}^{*}_{k}\exp{(j2\pi\Omega^{k}_{D}\zeta_{p})}\exp{(j2\pi\Omega^{k}_{\theta}
    (\alpha_{n}+\beta_{m}))}\nonumber\\
    &\times\frac{1}{\sqrt{N}}\sum_{t=0}^{N-1}\exp{\left(j2\pi\left(\Omega^{k}_{R}-\frac{l}{N}\right)t\right)}+W_{n,m,p}[l],\label{eqn:range DFT prev}
\end{align}
\normalsize
for $0\leq l\leq N-1$ and $W_{n,m,p}[l]=(1/\sqrt{N})\sum_{t=0}^{N-1}w_{n,m,p}[t]\exp{(-j2\pi lt/N)}$.

Now, we can approximate the sum of $M$ exponents $g(x|\overline{x})= \frac{1}{\sqrt{M}}\sum_{q=0}^{M-1}e^{j(x-\overline{x})q\omega}$ for given constants $\overline{x}$ and $\omega$ as
\par\noindent\small
\begin{align}
    |g(x|\overline{x})|\approx\begin{cases}\sqrt{M},&|x-\overline{x}|\leq\pi/M\omega\\0,&|x-\overline{x}|>\pi/M\omega\end{cases}.\label{eqn:focusing approx}
\end{align}
\normalsize
This approximation implies that in the focus zone $|x-\overline{x}|\leq\pi/M\omega$, the $M$ exponents are coherently integrated while the signal outside the focus zone is severely attenuated. This focusing operation was applied across pulses in \cite{bar2014sub} to reduce the joint delay-Doppler estimation problem to delay-only estimation in a pulsed-Doppler radar. Particularly, target returns from different pulses were combined into a single high-SNR pulse via time-shifting and modulation. Applying approximation \eqref{eqn:focusing approx} with windowing restricted the focus to a narrow Doppler band, within which delay estimation became feasible. In contrast, our formulation in \eqref{eqn:range DFT prev} inherently yields a sum of exponentials in the sampled IF signal’s DFT, without any windowing. We now show that applying approximation \eqref{eqn:focusing approx} to \eqref{eqn:range DFT prev} directly concentrates target returns from the same range into a single bin, independent of their Doppler and angular frequencies ($\Omega^{k}_{D}$ and $\Omega^{k}_{\theta}$).

Using the focusing approximation for the sum $\frac{1}{\sqrt{N}}\sum_{t=0}^{N-1}\exp{\left(j2\pi\left(\Omega^{k}_{R}-\frac{l}{N}\right)t\right)}$ in \eqref{eqn:range DFT prev}, we have
\par\noindent\small
\begin{align}
    Y_{n,m,p}[l]&=\sum_{k'=1}^{K'}\widetilde{a}^{*}_{k'}\sqrt{N}\exp{(j2\pi\Omega^{k'}_{D}\zeta_{p})}\exp{(j2\pi\Omega^{k'}_{\theta}
    (\alpha_{n}+\beta_{m}))}\nonumber\\
    &\;\;\;+W_{n,m,p}[l],\label{eqn:range DFT}
\end{align}
\normalsize
where $\{\widetilde{a}_{k'},\Omega^{k'}_{R},\Omega^{k'}_{D},\Omega^{k'}_{\theta}\}_{1\leq k'\leq K'}$ represents the subset of $K'$ target's that satisfy $|\Omega^{k'}_{R}-l/N|\leq 1/2N$, i.e., $K'$ targets are present in $l$-th range bin. Substituting $N=f_{s}T_{c}$ and $\Omega^{k'}_{R}=\gamma\tau^{R}_{k'}/f_{s}$, the condition $|\Omega^{k'}_{R}-l/N|\leq 1/2N$ simplifies to
$|\tau^{R}_{k'}-l/\gamma T_{c}|\leq 1/2\gamma T_{c}$. Furthermore, for practical FMCW radar systems, $1/2\gamma T_{c}$ is small such that
\par\noindent\small
\begin{align}
   \tau^{R}_{k'}\approx l/\gamma T_{c}.\label{eqn:range delay estimate} 
\end{align}
\normalsize
The target returns with range delay \eqref{eqn:range delay estimate} are coherently integrated and result in a (magnitude) peak at the $l$-th DFT bin. We can identify these peaks in $Y_{n,m,p}[l]$ using threshold detection and obtain the range estimate corresponding to $l'$-th detected bin as $R'=cl'/2\gamma T_{c}$.

\textit{Binary integration:} So far, we have considered range estimation from a given $p$-th chirp and $(n,m)$-th virtual array channel. Similarly, the range estimates are computed independently for all $P$ chirps and $N_{T}N_{R}$ channels using the radar measurements $\{y_{n,m,p}[t]\}_{1\leq n\leq N_{T},1\leq m\leq N_{R},1\leq p\leq P}$. In binary integration, these detected ranges are then filtered for false alarms and missed detections across all chirps and channels using a majority rule, i.e., only the ranges detected in the majority of measurements are considered valid. Note that binary integration can also be interpreted as a statistical smoothing technique that mitigates the impact of noise. However, unlike traditional smoothing used in super-resolution methods such as MUSIC, the primary objective here is to suppress false alarms in low-SNR scenarios. In contrast, MUSIC employs smoothing to decorrelate coherent sources and achieve high-resolution estimation at moderate to high SNRs. Moreover, conventional smoothing methods typically involve higher computational complexity compared to binary integration.
\begin{figure}
  \centering
  \includegraphics[width = 1.0\columnwidth]{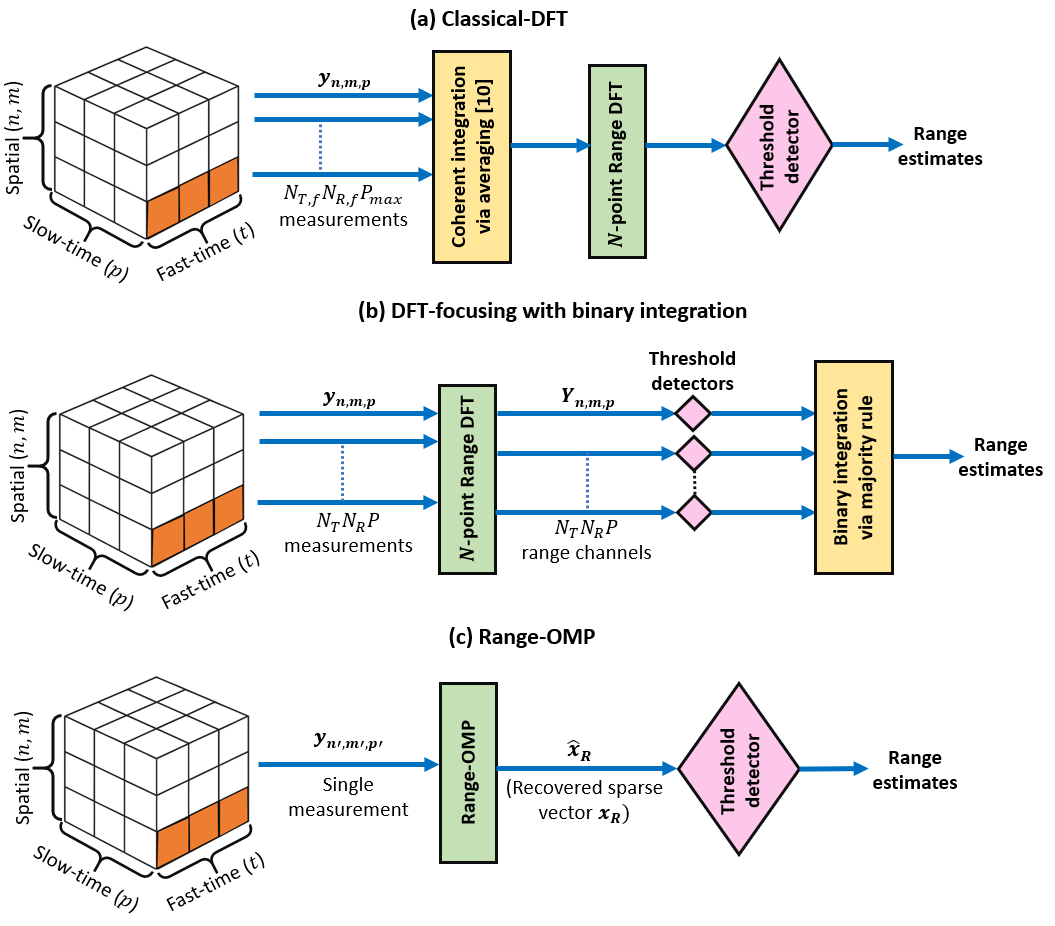}
  \caption{Graphical illustration of range estimation using (a) classical-DFT, (b) DFT-focusing with binary integration, and (c) Range-OMP methods.}
 \label{fig:range schematic}
\end{figure}

\begin{remark}[Coherent and binary integration]\label{remark:coherent binary}
    Classical-DFT range processing also involves peak detection in the DFT of the sampled IF signal. However, in classical processing, the measurements from all chirps and array channels are processed coherently to average out the noise term. At low SNRs, the method fails to detect range bins where the noise is significant. Contrarily, in binary integration, a missed range in one measurement can be detected in other measurements, reducing the missed detection probability, as demonstrated numerically in Section~\ref{subsec:range numerical}. Note that the classical-DFT method utilizes full measurements from an $N_{T,f}\times N_{R,f}$ ULA and $P_{max}$ chirps, whereas our DFT-focusing approach relies on sparse measurements, leading to reduced computational complexity, as detailed in Table~\ref{tbl:comparison}. However, our joint Doppler-angle estimation procedure proposed later can also be trivially applied to coherently processed measurements.
\end{remark}

\textbf{2) Range-OMP:} In \eqref{eqn:range delay estimate}, the range resolution is $c/2\gamma T_{c}$ where $\gamma T_{c}$ is the total bandwidth swept by the LFM chirp. This resolution is the same as the classical-DFT \cite{richards2014fundamentals}, which requires  large bandwidth and heavy computation. To this end, we propose an alternative OMP-based range recovery method.

Consider the measurements $y_{1,1,1}[t]$ from the first chirp and the first virtual array channel. Denote $x_{k}\doteq\widetilde{a}^{*}_{k}\exp{(j2\pi\Omega^{k}_{D}\zeta_{p})}\exp{(j2\pi\Omega^{k}_{\theta}(\alpha_{n}+\beta_{m}))}$ and stack the samples $y_{1,1,1}[t]$ for $0\leq t\leq N-1$ in a $N\times 1$ vector $\mathbf{y}_{R}$. Now, define the $K\times 1$ vector $\widetilde{\mathbf{x}}_{R}=[x_{1},\hdots,x_{K}]^{T}$ such that \eqref{eqn:3D mixture} yields
\par\noindent\small
\begin{align}
   \mathbf{y}_{R}=\widetilde{\mathbf{A}}(\bm{\Omega}_{R})\widetilde{\mathbf{x}}_{R}+\mathbf{w}_{R},\label{eqn:range omp sub}
\end{align}
\normalsize
where the $N\times K$ matrix $\widetilde{\mathbf{A}}(\bm{\Omega}_{R})=[\mathbf{a}(\Omega^{1}_{R}),\hdots,\mathbf{a}(\Omega^{K}_{R})]$ with the $k$-th column $\mathbf{a}(\Omega^{k}_{R})\doteq[1,\exp{(j2\pi\Omega^{k}_{R})},\hdots,\exp{(j2\pi\Omega^{k}_{R}(N-1))}]^{T}$. Here, $\mathbf{w}_{R}$ represents the $N\times 1$ stacked noise vector. Each column $\mathbf{a}(\Omega_{R})$ of matrix $\widetilde{\mathbf{A}}$ is parameterized by $\Omega_{R}$ and is the equivalent steering vector in the (range) beat frequency domain. Furthermore, $\mathbf{y}_{R}$ in \eqref{eqn:range omp sub} represents full measurements in the fast-time domain sampled at frequency $f_{s}$.

We choose a grid of $G_{R}$ points $\{\omega_{g}:1\leq g\leq G_{R}\}$ of the possible target beat frequencies $\Omega_{R}$ (or equivalently ranges) with $G_{R}\gg K$ and negligible discretization errors. Substituting these grid points in $\mathbf{a}(\cdot)$, we construct an over-complete $N\times G_{R}$ measurement matrix $\mathbf{A}=[\mathbf{a}(\omega_{1}),\hdots,\mathbf{a}(\omega_{G_{R}})]$. Then, \eqref{eqn:range omp sub} becomes
\par\noindent\small
\begin{align}   \mathbf{y}_{R}=\mathbf{A}\mathbf{x}_{R}+\mathbf{w}_{R},\label{eqn:range omp}
\end{align}
\normalsize
where the $G_{R}\times 1$ vector $\mathbf{x}_{R}$ contains the target ranges and unknown coefficients $\{x_{k}\}$. In particular, a non-zero element of $\mathbf{x}_{R}$ represents a target present at the range corresponding to the grid point. Since $K\ll G_{R}$, $\mathbf{x}_{R}$ is a sparse vector and range estimation reduces to determining $\textrm{supp}(\mathbf{x}_{R})$ given $\mathbf{y}_{R}$ and $\mathbf{A}$. To this end, we can employ any 1D-OMP or BP recovery algorithms described in Section~\ref{subsec:CS summary}. In particular, we use OMP in Section~\ref{sec:numericals}. Note that our Range-OMP employs the standard OMP algorithm but utilizes full measurements in the radar’s fast-time domain, rather than compressed ones. It is used here to achieve high-resolution range estimation with reduced computational cost compared to conventional DFT-based methods.

\begin{remark}[Range-OMP and DFT-focusing]\label{remark:range dft omp}
Range-OMP in \eqref{eqn:range omp} considers measurements only from one chirp and one array channel. This is in contrast to DFT-focusing, wherein all the measurements are integrated (binary or coherently) to obtain valid target ranges making them computationally expensive. Furthermore, Range-OMP provides superior range resolution compared to DFT-focusing, determined by the choice of grid $\omega_{1\leq g\leq G_{R}}$ and not by parameters $\gamma$ and $T_{c}$. However, the $\textrm{supp}(\mathbf{x}_{R})$'s recovery still depends on the matrix $\mathbf{A}$ such that the recovery probability is low if the grid points are too close. Note that while binary integration can also be employed after Range-OMP to suppress false targets at low SNRs, this step becomes unnecessary in high-SNR scenarios.
\end{remark}

\begin{remark}[Computational complexity]\label{remark:complexity_range}
The computational complexity of DFT-focusing depends on the number of DFT points $N$, the number of virtual array channels $N_{T}N_{R}$, and the number of pulses $P$ per CPI. Specifically, computing the complex DFT has a complexity of  $\mathcal{O}(N_{T}N_{R}PN\log(N))$\cite{cooley1965algorithm}, followed by binary integration with complexity $\mathcal{O}(N_{T}N_{R}PN)$\cite{richards2014fundamentals}. Although this matches the complexity order of conventional DFT-based range estimation using full measurements (Table~\ref{tbl:comparison}), our DFT-focusing approach incurs a lower computational burden due to the significantly reduced number of measurements $N_{T}N_{R}P$, enabled by the random SLA and sparse chirps. In contrast, our Range-OMP uses measurements from only a single channel and chirp, and follows the standard OMP framework. Its complexity is determined by the number of targets $K$, the size of the range grid $G_{R}$, and the number of DFT points $N$, resulting in a complexity of $\mathcal{O}(KNG_{R})$\cite{tropp2007signal}. Moreover, using a finer range grid (larger $G_{R}$) enables better range resolution while incurring lower computational cost than DFT-based methods, as the number of targets $K$ is much smaller than the number of measurements $N_{T}N_{R}P$.
\end{remark}

\subsection{Joint Doppler-angle estimation}\label{subsec:angle-doppler processing}
Consider a range detected at the $l'$-th DFT bin while the corresponding DFT coefficients $Y_{n,m,p}[l']$ are given by \eqref{eqn:range DFT} for $1\leq p\leq P$, $1\leq n\leq N_{T}$ and $1\leq m\leq N_{R}$. The following joint Doppler-angle estimation method processes DFT measurements given by \eqref{eqn:range DFT} at a particular range bin. In the Range-OMP case, the target ranges are estimated at a higher resolution grid $\omega_{1\leq g\leq G_{R}}$, but then converted to the associated DFT bin and the measurements $Y_{n,m,p}[l']$ are computed. At the $l'$-th range bin, we assume that $K'$ targets are present such that
\par\noindent\small
\begin{align}
    Y_{n,m,p}[l']&=\sum_{k'=1}^{K'}\widetilde{a}^{*}_{k'}\sqrt{N}\exp{(j2\pi\Omega^{k'}_{D}\zeta_{p})}\exp{(j2\pi\Omega^{k'}_{\theta}
    (\alpha_{n}+\beta_{m}))}\nonumber\\
    &\;\;\;+W_{n,m,p}[l'].\label{eqn:angle-Doppler measurments}
\end{align}
\normalsize
Note that the exponential terms with the Doppler and angular frequencies are separable in $Y_{n,m,p}[l']$.

We stack the measurements $Y_{n,m,p}[l']$ for all $1\leq n\leq N_{T}$ and $1\leq m\leq N_{R}$ in an $N_{T}N_{R}\times 1$ vector $\mathbf{y}_{p}$ for the $p$-th chirp. Further, define the $N_{T}N_{R}\times P$ matrix $\mathbf{Y}=[\mathbf{y}_{1},\hdots,\mathbf{y}_{P}]$ and the steering vectors in the Doppler and angular frequency domains as $\mathbf{b}(\Omega_{D})$ and $\mathbf{c}(\Omega_{\theta})$, respectively, given by
\par\noindent\small
\begin{align}
    &\mathbf{b}(\Omega_{D})\doteq[\exp{(j2\pi\Omega_{D}\zeta_{1})},\hdots,\exp{(j2\pi\Omega_{D}\zeta_{P})}]^{T},\label{eqn:doppler steering vector}\\
    &\mathbf{c}(\Omega_{\theta})\doteq[\exp{(j2\pi\Omega_{\theta}(\alpha_{1}+\beta_{1}))},\hdots,\exp{(j2\pi\Omega_{\theta}(\alpha_{N_{T}}+\beta_{N_{R}}))}]^{T}.\label{eqn:angle steering vector}
\end{align}
\normalsize
Similarly, we define the $K'\times K'$ matrix $\widetilde{\mathbf{Z}}\doteq\textrm{diag}(z_{1},z_{2},\hdots,z_{K'})$ where $z_{k'}=\widetilde{a}^{*}_{k'}\sqrt{N}$. Then, using \eqref{eqn:angle-Doppler measurments}, we obtain
\par\noindent\small
\begin{align}
    \mathbf{Y}=\widetilde{\mathbf{C}}(\bm{\Omega}_{\theta})\widetilde{\mathbf{Z}}\widetilde{\mathbf{B}}^{T}(\bm{\Omega}_{D})+\mathbf{W},\label{eqn:angle-Doppler 2D-CS sub}
\end{align}
\normalsize
where the $N_{T}N_{R}\times K'$ matrix $\widetilde{\mathbf{C}}(\bm{\Omega}_{\theta})=[\mathbf{c}(\Omega^{1}_{\theta}),\hdots,\mathbf{c}(\Omega^{K'}_{\theta})]$ and the $P\times K'$ matrix $\widetilde{\mathbf{B}}(\bm{\Omega}_{D})=[\mathbf{b}(\Omega^{1}_{D}),\hdots,\mathbf{b}(\Omega^{K'}_{D})]$. Again, $\mathbf{W}$ represents the $N_{T}N_{R}\times P$ stacked noise matrix.

Our goal is to recover $\bm{\Omega_{\theta}}=\{\Omega^{1}_{\theta},\hdots,\Omega^{K'}_{\theta}\}$ and $\bm{\Omega_{D}}=\{\Omega^{1}_{D},\hdots,\Omega^{K'}_{D}\}$ from $\mathbf{Y}$ with a small number of transmitted chirps and antenna elements. To this end, we exploit the sparseness of the target scene. We choose grids of $G_{D}$ points $\rho_{1\leq g\leq G_{D}}$ and $G_{\theta}$ points $\phi_{1\leq g\leq G_{\theta}}$ of the possible target Doppler and angular frequencies, respectively, with both $G_{D},G_{\theta}\gg K$ and negligible discretization errors. Note that from $\Omega_{D}=2\nu T_{c}/\lambda$ and $\Omega_{\theta}=A\sin{(\theta)}/2\lambda$, the grids $\rho_{1\leq g\leq G_{D}}$ and $\phi_{1\leq g\leq G_{\theta}}$ are equivalently defined on possible target velocities and AOAs, respectively. Substituting these grid points in $\mathbf{b}(\cdot)$ and $\mathbf{c}(\cdot)$, we construct the $P\times G_{D}$ measurement matrix $\mathbf{B}=[\mathbf{b}(\rho_{1}),\hdots,\mathbf{b}(\rho_{G_{D}})]$ and the $N_{T}N_{R}\times G_{\theta}$ measurement matrix $\mathbf{C}=[\mathbf{c}(\phi_{1}),\hdots,\mathbf{c}(\phi_{G_{\theta}})]$. Finally, \eqref{eqn:angle-Doppler 2D-CS sub} can be expressed as
\par\noindent\small
\begin{align}
    \mathbf{Y}=\mathbf{C}\mathbf{Z}\mathbf{B}^{T}+\mathbf{W},\label{eqn:angle-Doppler 2D}
\end{align}
\normalsize
where the unknown $G_{\theta}\times G_{D}$ matrix $\mathbf{Z}$ contains the coefficients $\{z_{k'}\}$ as well as the target velocities and AOAs information. Particularly, a non-zero element in $\mathbf{Z}$ represents a target with its AOA and velocity equal to that of the corresponding point in the 2D grid $\{(\phi_{i_{1}},\rho_{i_{2}}):1\leq i_{1}\leq G_{\theta},1\leq i_{2}\leq G_{D}\}$. Since $K'\ll G_{D}G_{\theta}$, the matrix $\mathbf{Z}$ is sparse, and hence, the desired velocities and AOAs can be recovered using small number of measurements $\mathbf{Y}$ given matrices $\mathbf{C}$ and $\mathbf{B}$. Note that the joint Doppler-angle estimation has now been reduced to determining $\textrm{supp}(\mathbf{Z})$. In \eqref{eqn:angle-Doppler 2D}, the measurement matrices $\mathbf{C}$ and $\mathbf{B}$ and hence, the recovery guarantees, depend on the choice of girds $\rho_{1\leq g\leq G_{D}}$ and $\phi_{1\leq g\leq G_{\theta}}$ as well as the number and (random) locations of the transmitted chirps ($\zeta_{1\leq p\leq P}$) and antenna elements ($\alpha_{1\leq n\leq N_{T}}$ and $\beta_{1\leq m\leq N_{R}}$).
    \begin{table}
    \caption{Computational complexity of joint Doppler-angle estimation ($N_{\text{int}}$ is the number of iterations)}
    \label{tbl:complexity}
    \centering
    \begin{tabular}{p{1.6cm}p{3.2cm}p{2.5cm}}
    \hline\noalign{\smallskip}
    \textbf{Method} & \textbf{Complexity} & \textbf{Core Operation}\\
    \noalign{\smallskip}
    \hline
    \noalign{\smallskip}
    Vectorized 1D-OMP & $\mathcal{O}(K N_{T} N_{R} P G_{\theta} G_{D})$ & Projection and least-squares\\
    Vectorized BP & $\mathcal{O}(G_{\theta}^{3} G_{D}^{3})$ & Linear programming (e.g., interior point)\\
    Vectorized LASSO & $\mathcal{O}(N_{\text{int}} N_{T} N_{R} P G_{\theta} G_{D})$ & Convex optimization (e.g., ISTA, FISTA)\\
    2D-OMP & $\mathcal{O}(K (N_{T} N_{R} + P) G_{\theta} G_{D})$ & Matrix projection and least-squares\\
    \noalign{\smallskip}
    \hline\noalign{\smallskip}
    \end{tabular}
    \end{table}

Recall that the 2D sparse recovery problem \eqref{eqn:angle-Doppler 2D} can be solved using either 1D vectorized CS or direct 2D-OMP methods. For the vectorized method, we denote $\mathbf{y}=\textrm{vec}(\mathbf{Y})$, $\mathbf{z}=\textrm{vec}(\mathbf{Z})$, $\mathbf{w}=\textrm{vec}(\mathbf{W})$ and $\mathbf{D}=\mathbf{B}\otimes\mathbf{C}$. The equivalent 1D sparse recovery model of \eqref{eqn:angle-Doppler 2D} is then
\par\noindent\small
\begin{align}
    \mathbf{y}=\mathbf{D}\mathbf{z}+\mathbf{w},\label{eqn:angle-Doppler 1D vectorized}
\end{align}
\normalsize
which is same as \eqref{eqn:sparse 1d problem} with $\mathbf{A}$ and $\mathbf{x}$ replaced by $\mathbf{D}$ and $\mathbf{z}$, respectively. In Section~\ref{sec:guarantees}, we consider the $l_{1}$-minimization-based solution of \eqref{eqn:angle-Doppler 1D vectorized} to derive the theoretical recovery guarantees while Section~\ref{sec:numericals} numerically compares 2D-OMP and vectorized 1D-OMP, BP, and LASSO methods. Notably, with fixed Doppler and angular grids, these methods exhibit similar detection performance and estimation accuracy. However, their computational complexities vary, as summarized in Table~\ref{tbl:complexity}. The complexity of all approaches primarily depends on the number of grid points $G_{D}$ and $G_{\theta}$, the number of virtual array channels $N_{T}N_{R}$, the number of chirps $P$, and the number of targets $K$, but the vectorized BP method incurs the highest computational cost. This is because $G_{D}G_{\theta}$ is typically set to a large value to achieve fine Doppler and angular resolution. In contrast, 2D-OMP achieves a computational advantage through direct matrix operations, offering a speedup factor of $N_{T}N_{R}P/(N_{T}N_{R}+P)$ over vectorized 1D-OMP\cite{fang20122d}.

\section{Recovery guarantees}\label{sec:guarantees}
As mentioned in Section~\ref{subsec:CS summary}, the $l_{1}$-minimization formulations are relaxations of the original CS recovery problem. Hence, the solutions obtained from \eqref{eqn:BP} or \eqref{eqn:LASSO} can be different from that of \eqref{eqn:CS l0 problem}. To this end, sufficient conditions to guarantee correct recovery through $l_1$ minimization methods are extensively studied\cite{foucart2013invitation,rauhut2010compressive,eldar2012compressed,hugel2014remote}. In this context, two kinds of recovery guarantees are defined:\\
\textbf{1) Uniform recovery:} Uniform recovery means that for a fixed instantiation of the random measurement matrix $\mathbf{D}$, all possible $K$-sparse vectors $\mathbf{z}$ are recovered with high probability. Uniform recovery is guaranteed if $\mathbf{D}$ satisfies the restricted isometry property (RIP) with high probability \cite{eldar2012compressed}. Furthermore, if the mutual coherence $\mu$ of matrix $\mathbf{D}$ is small, then $\mathbf{D}$ satisfies RIP where
\par\noindent\small
\begin{align}
    \mu\doteq\textrm{max}_{i\neq l}\frac{|\mathbf{D}_{:,i}^{H}\mathbf{D}_{:,l}|}{\|\mathbf{D}_{:,i}\|_{2}\|\mathbf{D}_{:,l}\|_{2}}\label{eqn:coherence}
\end{align}
\normalsize
\textbf{2) Non-uniform recovery:} Consider a given $K$-sparse vector $\mathbf{z}_{s}$ and a measurement matrix $\mathbf{D}$ drawn at random independent of $\mathbf{z}_{s}$. Non-uniform recovery guarantees provide conditions under which $\mathbf{z}_{s}$ is recovered with high probability. Non-uniform recovery follows if $\mathbf{D}$ satisfies isotropy property with high probability \cite{candes2011probabilistic}, i.e., the components of each row of $\mathbf{D}$ have unit variance and are uncorrelated such that for every $i$,
\par\noindent\small
\begin{align}
    \mathbb{E}[\mathbf{D}_{i,:}^{H}\mathbf{D}_{i,:}]=\mathbf{I}.\label{eqn:isotropy}
\end{align}
\normalsize
Note that uniform recovery implies non-uniform recovery, but the converse is not true.

In the following, we consider our sparse chirps and SLA radar setup, described in Section~\ref{sec:system model}, and derive sufficient conditions for both uniform and non-uniform recovery for the vectorized joint Doppler-angle estimation problem \eqref{eqn:angle-Doppler 1D vectorized}. Particularly, we provide sufficient conditions on distributions ($\mathcal{P}_{p}$, $\mathcal{P}_{\alpha}$ and $\mathcal{P}_{\beta}$), number of chirps and antenna elements ($P$, $N_{T}$ and $N_{R}$), and grids ($\rho_{1\leq g\leq G_{D}}$ and $\phi_{1\leq g\leq G_{\theta}}$) such that the RIP and isotropy properties are satisfied with high probability. Section~\ref{subsec:example} then discusses practical distribution and grid choices that satisfy these conditions.

Consider the vectorized model \eqref{eqn:angle-Doppler 1D vectorized}. Since $\mathbf{D}=\mathbf{B} \otimes \mathbf{C}$, any column of $\mathbf{D}$ is represented as $\mathbf{D}_{:,G_{\theta}(i_{1}-1)+i_{2}}=\mathbf{B}_{:,i_{1}}\otimes\mathbf{C}_{:,i_{2}}$ where $1\leq i_{1}\leq G_{D}$ and $1\leq i_{2}\leq G_{\theta}$. Similarly, any row of $\mathbf{D}$ is represented by $\mathbf{D}_{N_{T}N_{R}(i_{1}-1)+i_{2},:}=\mathbf{B}_{i_{1},:}\otimes\mathbf{C}_{i_{2},:}$ where $1\leq i_{1}\leq P$ and $1\leq i_{2}\leq N_{T}N_{R}$. Further, consider a Doppler grid point $\rho_{i}$. Substituting $\Omega_{D}=2\nu T_{c}/\lambda$ in \eqref{eqn:doppler steering vector}, we obtain the $i$-th column of matrix $\mathbf{B}$ corresponding to $\rho_{i}$ as $\mathbf{b}(\rho_{i})=[\exp{(j\frac{4\pi T_{c}}{\lambda}\rho_{i}\zeta_{1})},\hdots,\exp{(j\frac{4\pi T_{c}}{\lambda}\rho_{i}\zeta_{P})}]^{T}$. Similarly, substituting $\Omega_{\theta}=A\sin{\theta}/2\lambda$ in \eqref{eqn:angle steering vector}, we obtain the $i$-th column of matrix $\mathbf{C}$ corresponding to angular grid point $\phi_{i}$ as \small$\mathbf{c}(\phi_{i})=[\exp{(j \frac{\pi A}{\lambda}\sin{\phi_{i}}(\alpha_{1}+\beta_{1}))},\hdots,\exp{(j \frac{\pi A}{\lambda}\sin{\phi_{i}}(\alpha_{N_{T}}+\beta_{N_{R}}))}]^{T}$.\normalsize\\ Define $\mathbf{Q}_{D}\doteq\mathbf{D}^{H}\mathbf{D}$, $\mathbf{Q}_{B}\doteq\mathbf{B}^{H}\mathbf{B}$ and $\mathbf{Q}_{C}\doteq\mathbf{C}^{H}\mathbf{C}$. From \eqref{eqn:coherence} and \eqref{eqn:isotropy}, we observe that the RIP and isotropy conditions for measurement matrix $\mathbf{D}$ are closely related to the statistics of matrix $\mathbf{Q}_{D}$. In particular, coherence $\mu$ is the maximum absolute value among normalized off-diagonal elements of $\mathbf{Q}_{D}$ while $\mathbb{E}[\mathbf{D}_{l,:}^{H}\mathbf{D}_{l,:}]=(1/PN_{T}N_{R})\mathbb{E}[\mathbf{Q}_{D}]$\footnote{trivially proved taking element-wise expectations and comparing matrix elements on both sides.} such that \eqref{eqn:isotropy} implies $\mathbb{E}[\mathbf{Q}_{D}]=PN_{T}N_{R} \mathbf{I}$. Also, we can trivially show that $\mathbf{Q}_{D}=\mathbf{Q}_{B}\otimes\mathbf{Q}_{C}$. For simplicity, denote $u^{D}_{i_{1},i_{2}}=\frac{4\pi T_{c}}{\lambda}(\rho_{i_{2}}-\rho_{i_{1}})$ and $u^{\theta}_{i_{1},i_{2}}=\frac{\pi A}{\lambda}(\sin{\phi_{i_{2}}}-\sin{\phi_{i_{1}}})$.

We define
\par\noindent\small
\begin{align}
\Gamma_{B}(u^{D}_{i_{1},i_{2}})&\doteq\frac{\mathbf{B}_{:,i_{1}}^{H}\mathbf{B}_{:,i_{2}}}{\|\mathbf{B}_{:,i_{1}}\|_{2}\|\mathbf{B}_{:,i_{2}}\|_{2}}=\frac{1}{P}\sum_{p=1}^{P}\exp{(ju^{D}_{i_{1},i_{2}}\zeta_{p})},\label{eqn:gamma b}\\
\Gamma_{C}(u^{\theta}_{i_{1},i_{2}})&\doteq\frac{\mathbf{C}_{:,i_{1}}^{H}\mathbf{C}_{:,i_{2}}}{\|\mathbf{C}_{:,i_{1}}\|_{2}\|\mathbf{C}_{:,i_{2}}\|_{2}}\nonumber\\
&=\frac{1}{N_{T}N_{R}}\sum_{n=1}^{N_{T}}\sum_{m=1}^{N_{R}}\exp{(ju^{\theta}_{i_{1},i_{2}}(\alpha_{n}+\beta_{m}))}\label{eqn:gamma c}.
\end{align}
\normalsize
Note that both $\Gamma_{B}(\cdot)$ and $\Gamma_{C}(\cdot)$ are random variables because of the randomly drawn transmitter and receiver locations (via $\alpha_{n}$ and $\beta_{m}$) and the randomly transmitted chirps (via $\zeta_{p}$).

\subsection{Uniform recovery}\label{subsec:uniform}
Consider the coherence of matrix $\mathbf{D}$ given as
\par\noindent\small
    \begin{align*}
    &\mu=\textrm{max}_{l_{1}\neq l_{2}}\frac{|\mathbf{D}_{:,l_{1}}^{H}\mathbf{D}_{:,l_{2}}|}{\|\mathbf{D}_{:,l_{1}}\|_{2}\|\mathbf{D}_{:,l_{2}}\|_{2}}\\    &=\textrm{max}_{(i_{1},j_{1})\neq(i_{2},j_{2})}\frac{|\mathbf{D}_{:,G_{\theta}(i_{1}-1)+j_{1}}^{H}\mathbf{D}_{:,G_{\theta}(i_{2}-1)+j_{2}}|}{\|\mathbf{D}_{:,G_{\theta}(i_{1}-1)+j_{1}}\|_{2}\|\mathbf{D}_{:,G_{\theta}(i_{2}-1)+j_{2}}\|_{2}}.
    \end{align*}
\normalsize
Substituting $\mathbf{D}_{:,G_{\theta}(i-1)+j}=\mathbf{B}_{:,i}\otimes\mathbf{C}_{:,j}$, we obtain
\par\noindent\small
    \begin{align*}
    \mu&=\textrm{max}_{\{(i_{1},j_{1})\neq(i_{2},j_{2})\}} \frac{|\mathbf{B}_{:,i_1}^{H}\mathbf{B}_{:,i_2}|}{\|\mathbf{B}_{:,i_1}\|_{2}\|\mathbf{B}_{:,i_2}\|_{2}}\times\frac{|\mathbf{C}_{:,j_1}^{H}\mathbf{C}_{:,j_2}|}{\|\mathbf{C}_{:,j_1}\|_{2}\|\mathbf{C}_{:,j_2}\|_{2}},\\
    &=\textrm{max}_{\{(i_{1},j_{1})\neq(i_{2},j_{2})} |\Gamma_{B}(u^{D}_{i_{1},i_{2}})|\times |\Gamma_{C}(u^{\theta}_{j_{1},j_{2}})|
    \end{align*}
\normalsize
Now, $\textrm{max}_{(i_{1},j_{1})\neq(i_{2},j_{2})}(\cdot)=\textrm{max}_{\{i_{1}\neq i_{2}\}\cup\{j_{1}\neq j_{2}\}}(\cdot)=\textrm{max}\{\textrm{max}_{i_{1}\neq i_{2}}(\cdot),\textrm{max}_{j_{1}\neq j_{2}}(\cdot)\}$. Furthermore, $|\Gamma_{C}(u^{\theta}_{j_{1},j_{2}})|\leq 1$ such that $\textrm{max}_{i_{1}\neq i_{2}}|\Gamma_{B}(u^{D}_{i_{1},i_{2}})|\times |\Gamma_{C}(u^{\theta}_{j_{1},j_{2}})|\leq\textrm{max}_{i_{1}\neq i_{2}}|\Gamma_{B}(u^{D}_{i_{1},i_{2}})|$. Similarly, $\textrm{max}_{j_{1}\neq j_{2}}|\Gamma_{B}(u^{D}_{i_{1},i_{2}})|\times |\Gamma_{C}(u^{\theta}_{j_{1},j_{2}})|\leq\textrm{max}_{j_{1}\neq j_{2}}|\Gamma_{C}(u^{\theta}_{j_{1},j_{2}})|$. Hence,
\par\noindent\small
    \begin{align}
    \mu\leq\textrm{max}\{\textrm{max}_{i_{1}\neq i_{2}}|\Gamma_{B}(u^{D}_{i_{1},i_{2}})|,\textrm{max}_{j_{1}\neq j_{2}}|\Gamma_{C}(u^{\theta}_{j_{1},j_{2}})|\}\label{eqn:mu expansion}
    \end{align}
\normalsize
\begin{lemma}\label{lemma:toeplitz}
    If $\rho_{1:G_{D}}$ is a uniform grid of Doppler frequencies, then matrix $\mathbf{Q}_{B}=\mathbf{B}^{H}\mathbf{B}$ is a Toeplitz matrix with all main diagonal elements equal to $P$. Similarly, if $\phi_{1:G_{\theta}}$ is an angular grid uniform in $\sin{\theta}$ domain, then $\mathbf{Q}_{C}=\mathbf{C}^{H}\mathbf{C}$ is a Toeplitz matrix with all main diagonal elements equal to $N_{T}N_{R}$.
\end{lemma}
\begin{proof}
See Appendix~\ref{App-lemma-toeplitz}.
\end{proof}

By the virtue of Lemma~\ref{lemma:toeplitz}, for uniform grids, the statistics of $\Gamma_{B}(\cdot)$ and $\Gamma_{C}(\cdot)$ (and consequently $\mu$) are characterized by considering only the first rows of $\mathbf{Q}_{B}$ and $\mathbf{Q}_{C}$, respectively, i.e., $\mathbf{B}_{:,1}^{H}\mathbf{B}_{:,i_{1}}$ for $i_{1}=2,3,\hdots, G_{D}$ and $\mathbf{C}_{:,1}^{H}\mathbf{C}_{:,i_{2}}$ for $i_{2}=2,3,\hdots, G_{\theta}$. Note that $\|\mathbf{B}_{:,i_{1}}\|_{2}=\sqrt{P}$ for all $i_{1}$ and $\|\mathbf{C}_{:,i_{2}}\|_{2}=\sqrt{N_{T}N_{R}}$ for all $i_{2}$.

In the following, we address two cases: (a) $N_{T}$ transmitters and $N_{R}$ receivers, where $\alpha_{1:N_{T}}$ and $\beta_{1:N_{R}}$ are independent, and (b) $N_{T}=N_{R}$ transceivers, where $\alpha_{n}=\beta_{n}$ for all $1\leq n\leq N_{T}$. Theorem~\ref{theorem:coherence} upper bounds the complementary cdf of $\mu$ for both these cases.

\begin{theorem}[Coherence of measurement matrix $\mathbf{D}$]\label{theorem:coherence}
    Consider the random sparse chirps and SLA radar setup of Section~\ref{sec:system model} wherein the transmitters' positions $\alpha_{1:N_{T}}$ and receivers' positions $\beta_{1:N_{R}}$ are drawn i.i.d. from even distributions $\mathcal{P}_{\alpha}$ and $\mathcal{P}_{\beta}$, respectively. The chirps' indices $\zeta_{1:P} \in \{0,1,\hdots,P_{max}-1\}$ are distinct and distributed as $\mathcal{P}_{p}$, which is symmetrical about $(P_{max}-1)/2$. The Doppler grid $\rho_{1:G_{D}}$ and angular grid $\phi_{1:G_{\theta}}$ are uniform in the velocity and $\sin{\theta}$ domains, respectively. Let the following assumptions hold true.\\
    \textbf{C1.} Distribution $\mathcal{P}_{p}$ and grid $\rho_{1:G_{D}}$ satisfy
        \par\noindent\small
        \begin{align}
            \Psi_{p}(u^{D}_{1,i})=\Psi_{p}(2u^{D}_{1,i})=0,\label{eqn:coherence v condition}
        \end{align}
        \normalsize
        for $i=2,3,\hdots,G_{D}$ where $\Psi_{p}(\cdot)$ denotes the characteristic function of $\mathcal{P}_{p}$.\\
    \textbf{C2.} Distributions $\mathcal{P}_{\alpha}$, $\mathcal{P}_{\beta}$ and grid $\phi_{1:G_{\theta}}$ satisfy
        \par\noindent\small
        \begin{align}
            \Psi_{\alpha}(u^{\theta}_{1,i})=\Psi_{\alpha}(2u^{\theta}_{1,i})=\Psi_{\beta}(u^{\theta}_{1,i})=\Psi_{\beta}(2u^{\theta}_{1,i})=0,\label{eqn:coherence alpha beta condition}
        \end{align}
        \normalsize
        for $i=2,3,\hdots,G_{\theta}$ where $\Psi_{\alpha}(\cdot)$ and $\Psi_{\beta}(\cdot)$ denote the characteristic functions of $\mathcal{P}_{\alpha}$ and $\mathcal{P}_{\beta}$, respectively.
        
    Then, for $0<\vartheta<1$, the coherence $\mu$ of matrix $\mathbf{D}$ in \eqref{eqn:angle-Doppler 1D vectorized} satisfies the following:\\
    \textbf{1)} If $\alpha_{1:N_{T}}$ and $\beta_{1:N_{R}}$ are independent, then
        \par\noindent\small
        \begin{align}
            \mathbb{P}(\mu>\vartheta)<1-&(1-e^{-\vartheta^{2} P})^{G_{D}-1}\nonumber\\
            &\times(1-2\vartheta\sqrt{N_{T}N_{R}}\;\mathcal{K}_{1}(2\vartheta\sqrt{N_{T}N_{R}}))^{G_{\theta}-1},\label{eqn:coherence bound independent case}
        \end{align}
        \normalsize
        where $\mathcal{K}_{1}(\cdot)$ denotes the modified Bessel function of second kind.\\
    \textbf{2)} For the $N_{T}=N_{R}$ transceivers case, $\alpha_{n}=\beta_{n}$ for all $1\leq n\leq N_{T}$ and
        \par\noindent\small
        \begin{align}
            \mathbb{P}(\mu>\vartheta)\leq 1-(1-e^{-\vartheta^{2}P})^{G_{D}-1}(1-e^{-N_{T}\vartheta})^{G_{\theta}-1}.\label{eqn:coherence transceiver case}
        \end{align}
        \normalsize
\end{theorem}
\begin{proof}
See Appendix~\ref{App-thm-coherence}.
\end{proof}

Using Theorem~\ref{theorem:coherence}, we can lower-bound the number of chirps $P$ and antenna elements $N_{T}N_{R}$ needed to ensure uniform recovery with high probability as follows.
\begin{theorem}[Uniform recovery guarantee]\label{theorem:uniform}
    Consider the random sparse chirps and SLA radar setup with the distributions $\mathcal{P}_{\alpha}$, $\mathcal{P}_{\beta}$ and $\mathcal{P}_{p}$ and grids $\rho_{1:G_{D}}$ and $\phi_{1:G_{\theta}}$ satisfying the conditions of Theorem~\ref{theorem:coherence}. Also, consider some $0<\epsilon<1$ with two positive constants $\epsilon_{1}$ and $\epsilon_{2}$   such that $\epsilon_{1}+\epsilon_{2}=\epsilon$ and\\
    \par\noindent\small
    \begin{align}
   P\geq\kappa_{1}\left(K-\frac{1}{2}\right)^{2}\log\left(\frac{G_{D}}{\epsilon_{1}}\right),\label{eqn:uniform P bound}
    \end{align}
    \normalsize
    and
    \par\noindent\small
    \begin{align}
    &N_{T}N_{R}\geq\kappa_{2}\left(K-\frac{1}{2}\right)^{2}\nonumber\\
    &\;\;\;\times\left(\log\left(\frac{G_{\theta}\sqrt{\pi}}{2\epsilon_{2}}\right)+\frac{1}{2}\log\left(2\log\left(\frac{G_{\theta}\sqrt{\pi}}{2\epsilon_{2}}\right)\right)\right)^{2},\label{eqn:uniform TR bound}
    \end{align}
    \normalsize
    when $\alpha_{1:N_{T}}$ and $\beta_{1:N_{R}}$ are independent, or
    \par\noindent\small
    \begin{align}
    N_{T}\geq\kappa_{3}\left(K-\frac{1}{2}\right)\log\left(\frac{G_{\theta}}{\epsilon_{2}}\right),\label{eqn:uniform T bound}
    \end{align}
    \normalsize
    when $\alpha_{n}=\beta_{n}$ for all $1\leq n\leq N_{T}$. Here, constants $\kappa_{1}\approx 18.69$, $\kappa_{2}\approx 4.67$ and $\kappa_{3}\approx 4.32$ while $K$ is the maximum number of targets present in a range bin. Then, for any $K$-sparse $\mathbf{z}$ measured by \eqref{eqn:angle-Doppler 1D vectorized} with $\|\mathbf{w}\|_{2}\leq\sigma$, the $l_{1}$-minimization solution $\hat{\mathbf{z}}$ satisfies $\|\hat{\mathbf{z}}-\mathbf{z}\|_{2}\leq\kappa_{4}\sigma$ with probability at least $1-\epsilon$, where $\kappa_{4}$ is a constant depending only on the RIP constant $\delta_{2K}$ of matrix $\mathbf{D}$.
\end{theorem}
\begin{proof}
See Appendix~\ref{App-thm-uniform}.
\end{proof}

Theorem~\ref{theorem:uniform} guarantees the exact recovery of any $K$-sparse signal in the noise-free $\sigma=0$ case with high probability. Note that various possible choices of $\epsilon_{1}$ and $\epsilon_{2}$ in \eqref{eqn:uniform P bound}-\eqref{eqn:uniform T bound} ensure that the number of chirps $P$ and antenna elements $N_{T}N_{R}$ can be flexibly adjusted to guarantee joint recovery of target velocities and AOAs. Additionally, $\Gamma_{C}(\cdot)$ in \eqref{eqn:gamma c} is closely related to the array pattern of our random SLA \cite{johnson1992array,rossi2013spatial} such that the bounds in \eqref{eqn:uniform TR bound} and \eqref{eqn:uniform T bound} signify the number of antenna elements necessary to control the peak sidelobes. Note that, similar to \cite{lo1964mathematical,rossi2013spatial}, the results in Theorems \ref{theorem:coherence}-\ref{theorem:uniform} are valid for sufficiently large measurements (due to the approximations made in their proofs) and primarily highlight the interdependence among measurements, grid points, and the number of targets for uniform recovery. However, uniform recovery investigates the worst-case scenario, but the average performance in practice is much better than predicted by Theorem~\ref{theorem:uniform}. Particularly, in our parameter estimation problem, we are more interested in recovering $\textrm{supp}(\mathbf{z})$ than the (exact) non-zero values in $\mathbf{z}$.
\begin{remark}[Dependence on number of targets]\label{remark:number of targets}
    From \eqref{eqn:uniform P bound}, we observe that the number of chirps $P$ may increase quadratically with the number of targets $K$ expected to be present in a range bin. The number of antenna elements ($N_{T}N_{R}$ in \eqref{eqn:uniform TR bound} or $N_{T}$ in \eqref{eqn:uniform T bound}) have a quadratic and linear dependence on $K$, respectively, for the independent transmitter-receiver and transceivers cases.
\end{remark}
\begin{remark}[Dependence on grid points]\label{remark:number of grid points}
    In Theorem~\ref{theorem:uniform}, both the number of chirps $P$ and antenna elements ($N_{T}N_{R}$ or $N_{T}$) increase logarithmically with the associated number of grid points ($G_{D}$ or $G_{\theta}$). This logarithmic dependence signifies that our proposed method enables high-resolution Doppler-angle estimation using only a small number of chirps and antenna elements. On the contrary, for traditional full MIMO radar arrays that transmit chirps for the entire CPI, the angular and velocity resolution depend, respectively, on the array aperture and CPI such that both the chirps and antenna elements scale linearly with grid points for high resolution. However, in our case, grid points $G_{D}$ and $G_{\theta}$ are not free variables because they need to satisfy \eqref{eqn:coherence v condition} and \eqref{eqn:coherence alpha beta condition}, respectively.
\end{remark}

\subsection{Non-uniform recovery}\label{subsec:non-uniform}
We now examine the isotropy condition \eqref{eqn:isotropy} for matrix $\mathbf{D}$ and non-uniform guarantees. Recall that any row of $\mathbf{D}$ is represented as $\mathbf{D}_{N_{T}N_{R}(i_{1}-1)+i_{2},:}=\mathbf{B}_{i_{1},:}\otimes\mathbf{C}_{i_{2},:}$ where $1\leq i_{1}\leq P$, $1\leq i_{2}\leq N_{T}N_{R}$. Now, any row of matrix $\mathbf{C}$ is \small$[\exp{(j\frac{\pi A}{\lambda}\sin{\phi_{1}}(\alpha_{n}+\beta_{m}))},\hdots,\exp{(j\frac{\pi A}{\lambda}\sin{\phi_{G_{\theta}}}(\alpha_{n}+\beta_{m}))}]$\;\normalsize
for some $1\leq n\leq N_{T}$ and $1\leq m\leq N_{R}$. Hence, rows of $\mathbf{C}$ and, consequently, $\mathbf{D}$ are generally not independent. Non-uniform recovery for such cases has been addressed in \cite{hugel2014remote}. Here, we first provide conditions to satisfy the isotropy property in Theorem~\ref{theorem:isotropy}. The non-uniform recovery guarantees then follow in Theorem~\ref{theorem:non-uniform}.

\begin{theorem}[Isotropy property for matrix $\mathbf{D}$]\label{theorem:isotropy}
Consider the random sparse chirps and SLA setup of Theorem~\ref{theorem:coherence}. Define random variable $\xi=\alpha+\beta$. The matrix $\mathbf{D}$ in \eqref{eqn:angle-Doppler 1D vectorized} satisfies  $\mathbb{E}[\mathbf{D}_{l,:}^{H}\mathbf{D}_{l,:}]=\mathbf{I}$ for every $1\leq l\leq PN_{T}N_{R}$ (isotropy property) iff the distributions $\mathcal{P}_{p}$, $\mathcal{P}_{\alpha}$ and $\mathcal{P}_{\beta}$, and uniform grids $\rho_{1:G_{D}}$ and $\phi_{1:G_{\theta}}$ satisfy
\par\noindent\small
\begin{align}
    \Psi_{p}(u^{D}_{1,i})=0,\;\;\;\textrm{and}\;\;\;\Psi_{\xi}(u^{\theta}_{1,j})=0,\label{eqn:condition for isotropy}
\end{align}
\normalsize
for $i=2,3,\hdots,G_{D}$ and $j=2,3,\hdots,G_{\theta}$, where $\Psi_{\xi}(\cdot)$ denotes the characteristic function of $\xi$.
\end{theorem}
\begin{proof}
    See Appendix~\ref{App-thm-isotropy}.
\end{proof}

\begin{theorem}[Non-uniform recovery guarantee]\label{theorem:non-uniform}
    Consider the random sparse chirps and SLA setup of Theorem~\ref{theorem:coherence}. Also, consider a given $K$-sparse signal $\mathbf{z}$ measured by \eqref{eqn:angle-Doppler 1D vectorized} with $\|\mathbf{w}\|_{2}\leq \sigma$. We assume that the distributions $\mathcal{P}_{p}$, $\mathcal{P}_{\alpha}$ and $\mathcal{P}_{\beta}$ and uniform grids $\rho_{1:G_{D}}$ and $\phi_{1:G_{\theta}}$ satisfy \eqref{eqn:condition for isotropy}. Consider some $0<\epsilon<1$ such that
    \par\noindent\small
\begin{align}
    PN_{T}N_{R}\geq \kappa_{5}K\log^{2}\left(\frac{\kappa_{6}G_{D}G_{\theta}}{\epsilon}\right)\label{eqn:PTR non-uniform}
\end{align}
\normalsize
where $\kappa_{5}\leq 2.87\times 10^{6}$ and $\kappa_{6}\leq 6$ are universal constants while $K$ denotes the maximum number of targets in a range bin. Then, the $l_{1}$ minimization solution $\hat{\mathbf{z}}$ satisfies
\par\noindent\small
\begin{align}
    \|\hat{\mathbf{z}}-\mathbf{z}\|_{2}\leq \kappa_{7}\sigma\sqrt{\frac{K}{PN_{T}N_{R}}}\label{eqn:non-uniform error bound}
\end{align}
\normalsize
with probability at least $1-\epsilon$, where constant $\kappa_{7}\leq 23.513$.
\end{theorem}
\begin{proof}
   Using the isotropy property, the non-uniform recovery of $K$-targets can be guaranteed by generalizing \cite[Theorem~2.1]{hugel2014remote} for the vectorized model \eqref{eqn:angle-Doppler 1D vectorized}. Particularly, we consider the exact $K$-sparse signal $\mathbf{z}$ and substitute the following: (a) number of rows of measurement matrix by $PN_{T}N_{R}$, (b) sparsity by $K$, and (c) number of columns of measurement matrix by $G_{D}G_{\theta}$.
\end{proof}

\begin{remark}[Uniform and non-uniform recovery]\label{remark:uniform vs nonuniform}
    Contrary to Theorem~\ref{theorem:uniform}, in non-uniform recovery, the total number of measurements $PN_{T}N_{R}$ scales linearly with number of targets $K$ and as $log^{2}(G_{D}G_{\theta})$ with the number of grid points. The inequality \eqref{eqn:PTR non-uniform} indicates that the number of chirps $P$ and antenna elements $N_{T}N_{R}$ can be flexibly adjusted to obtain high resolutions jointly in Doppler and AOA domains. Note that $G_{D}$ and $G_{\theta}$ are still not free variables since they need to satisfy \eqref{eqn:condition for isotropy}. However, unlike Theorem~\ref{theorem:uniform} where the error does not depend on the number of measurements, R.H.S of \eqref{eqn:non-uniform error bound} in non-uniform recovery is an explicit function of $PN_{T}N_{R}$. Again, Theorem~\ref{theorem:non-uniform} guarantees exact recovery for the noise-free $\sigma=0$ case.
\end{remark}

Intuitively, bounds \eqref{eqn:uniform P bound}-\eqref{eqn:uniform T bound} from Theorem~\ref{theorem:uniform} and bound \eqref{eqn:PTR non-uniform} from Theorem~\ref{theorem:non-uniform} provide the sufficient number of measurements necessary to control the off-diagonal elements of the measurement matrix $\mathbf{D}$, thereby ensuring high-probability recovery. In \eqref{eqn:coherence bound independent case}-\eqref{eqn:coherence transceiver case}, the coherence $\mu$, also viewed as the peak sidelobe of the array pattern\cite{rossi2013spatial}, provides a bound on the likelihood of this peak exceeding the threshold $\vartheta$. Meanwhile, the isotropy condition $\mathbb{E}[\mathbf{D}_{l,:}^{H}\mathbf{D}_{l,:}]=\mathbf{I}$ resembles the aperture condition in \cite{hugel2014remote} and serves to regulate the variance of the non-diagonal elements of the matrix $\mathbf{Q}_{D}$. Finally, conditions \eqref{eqn:coherence v condition}-\eqref{eqn:coherence alpha beta condition} and \eqref{eqn:condition for isotropy} imply that the grid points should be located at the zeros of the corresponding characteristic functions. We also note that Theorems~\ref{theorem:coherence}-\ref{theorem:non-uniform} provide sufficient—though not necessary—conditions for ensuring high-probability recovery. In practice, the target’s velocity and AOA in \eqref{eqn:angle-Doppler 1D vectorized} are estimated from the support of the sparse vector $\mathbf{z}$, rather than from the exact values of its non-zero entries. Consequently, in our experiments in Section~\ref{sec:numericals}, the proposed methods can match the detection performance of classical-DFT and MUSIC techniques, even when using only half the number of chirps and antenna elements.

\subsection{Practical example}\label{subsec:example}
Proposition~\ref{prop:example} provides a practical example of distributions $\mathcal{P}_{p}$, $\mathcal{P}_{\alpha}$ and $\mathcal{P}_{\beta}$ and grids $\rho_{1:G_{D}}$ and $\phi_{1:G_{\theta}}$ that satisfy \eqref{eqn:coherence v condition}-\eqref{eqn:coherence alpha beta condition}.

\begin{proposition}\label{prop:example}
    Consider $A_{T}=A_{R}=A/2$. Then,\\
    \textbf{1)} If $\mathcal{P}_{\alpha}$ and $\mathcal{P}_{\beta}$ are $\mathcal{U}[-1/2,1/2]$ and $\phi_{1:G_{\theta}}$ is a uniform grid in $\sin{\theta}$ domain with a spacing of $2\lambda/A$, then \eqref{eqn:coherence alpha beta condition} is satisfied.\\
    \textbf{2)} If $\mathcal{P}_{p}$ is $\mathcal{U}\{0,1,\hdots,P_{max}-1\}$ and $\rho_{1:G_{D}}$ is a uniform grid of spacing $\lambda/(2P_{max}T_{c})$, then \eqref{eqn:coherence v condition} is satisfied.
\end{proposition}
\begin{proof}
See Appendix~\ref{App-prop-example}.
\end{proof}

In Proposition~\ref{prop:example}, if $\phi_{1:G_{\theta}}$ covers the entire $[-\pi/2,\pi/2]$ interval, then $G_{\theta}=(A/\lambda)+1$ where $A/\lambda$ is the normalized array aperture. However, in practice, the radar focuses on a specific angular sector in each CPI. Similarly, denote $v_{max}=\lambda/4T_{c}$ as the maximum target velocity that can be detected by the given radar system. If $\rho_{1:G_{D}}$ covers $[-v_{max},v_{max}]$ interval, then $G_{D}=P_{max}+1$. Interestingly, the velocity resolution is then $\lambda/(2P_{max}T_{c})$, same as the traditional radar transmitting $P_{max}$ chirps. However, the proposed framework achieves this resolution by transmitting only $P<P_{max}$ chirps.
\begin{figure*}
\centering
\begin{subfigure}{0.23\textwidth}
    \includegraphics[width=\textwidth]{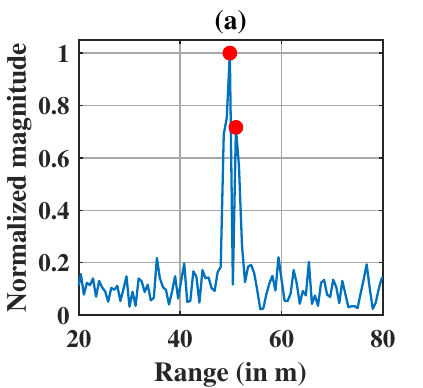}
\end{subfigure}
\hfill
\begin{subfigure}{0.23\textwidth}
    \includegraphics[width=\textwidth]{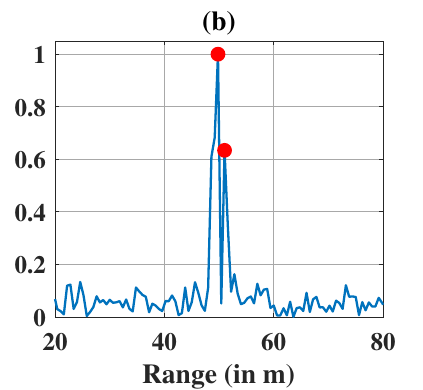}
\end{subfigure}
\hfill
\begin{subfigure}{0.23\textwidth}
    \includegraphics[width=\textwidth]{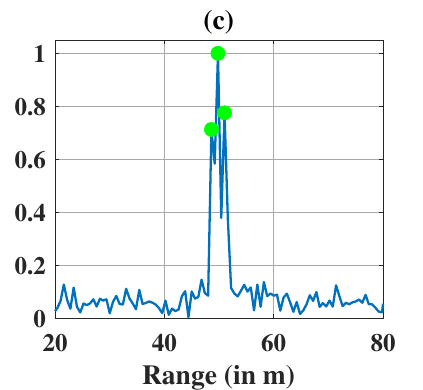}
\end{subfigure}
\hfill
\begin{subfigure}{0.23\textwidth}
    \includegraphics[width=\textwidth]{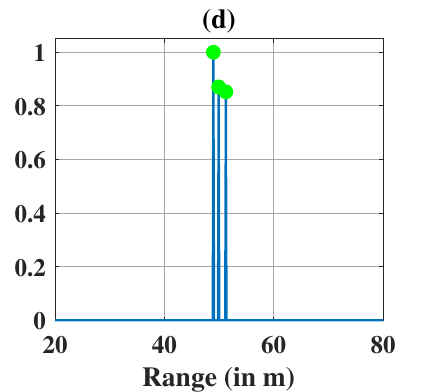}
\end{subfigure}
\caption{Range channels for DFT with (a) coherent integration, and (b) Channel-A ($y_{1,1,5}$) and (c) Channel-B ($y_{2,2,6}$) for binary integration; and (d) Range-OMP.}
\label{fig:range channels}
\end{figure*}

Finally, it is trivial to show that the non-uniform recovery condition \eqref{eqn:condition for isotropy} can be satisfied if $\mathcal{P}_{p}$ is chosen as a discrete uniform distribution and either one of $\mathcal{P}_{\alpha}$ or $\mathcal{P}_{\beta}$ is uniformly distributed. In this case, if $\mathcal{P}_{\alpha}$ is uniform, the receivers' locations $\beta_{1:N_{R}}$ can be chosen as a deterministic function of $\alpha_{1:N_{T}}$ (e.g., $\alpha_{n}=\beta_{n}$) and vice versa. Intuitively, both \eqref{eqn:coherence v condition}-\eqref{eqn:coherence alpha beta condition} and \eqref{eqn:condition for isotropy} impose that the grid points are placed at the zeros of the respective characteristic functions.

\section{Simulation results}\label{sec:numericals}
We compare the detection and estimation performance of our proposed methods with the classical-DFT and MUSIC techniques. Section~\ref{subsec:range numerical} considers various range estimation methods while Section~\ref{subsec:overall numerical} focuses on the overall performance and time-complexity of various  methods. The simulation parameters are listed in Table~\ref{tbl:parameters}. For SLA, $\alpha_{n},\beta_{m}\sim\mathcal{U}[-1/2,1/2]$ for all $1\leq n\leq N_{T}$ and $1\leq m\leq N_{R}$. For ULA, we consider the MIMO array setup of \cite{belfiori20122d} wherein two transmitters are placed on either side of the array with $\lambda$ spacing, and the receivers are placed between them with $0.5\lambda$ and $0.25\lambda$ spacings with one another and the transmitters, respectively. This arrangement results in $20$ unique virtual array channels. Similarly, for the sparsely transmitted chirps, $\zeta_{1\leq p\leq P}$ are chosen uniformly and independently from $\{0,1,\hdots,P_{max}-1\}$ without replacement, while the classical FMCW setup transmits all $P_{max}$ chirps. The target gains are generated as $a_{k}=\exp{(j\psi_{k})}$ with $\psi_{k}\sim\mathcal{U}[0,2\pi)$ (without any path loss) while the noise $w_{n,m,p}[t]\sim \mathcal{CN}(\mathbf{0},\sigma^{2}I)$, i.i.d. across all chirps and virtual array channels. Note that the SNR is $-10\log_{10}(\sigma^{2})$\cite{rossi2013spatial}. In all experiments, both the random SLA and ULA are assumed to be perfectly calibrated. In Appendix~\ref{App-calibration}, we demonstrate that the proposed methods with random SLAs exhibit similar sensitivity to phase and delay calibration errors as the subspace-based MUSIC approach with ULA, which is higher than the classical-DFT.
 
\subsection{DFT-focusing and Range-OMP}\label{subsec:range numerical}
Consider three close range targets with ranges $R_{1}=48.8$m, $R_{2}=50.0$m and $R_{3}=51.2$m for the full measurements radar setup. For the parameters given in Table~\ref{tbl:parameters}, the range resolution for the DFT processing is $0.6$m. In Range-OMP, we choose a uniform grid $\omega_{1\leq g\leq G_{R}}$ covering ranges $1.2$m to $120$m with resolution $0.12$m. Both the binary and coherent integration process all $20\times 32$ measurements $y_{n,m,p}$, but binary integration applies a $1/3$-rd majority rule to filter valid target ranges. Contrarily, Range-OMP considers only $y_{1,1,1}$ channel.

Fig.~\ref{fig:range channels} shows the normalized range channels for DFT (with coherent and binary integration) and Range-OMP at $-20$dB SNR. Fig.~\ref{fig:range channels}a \& b indicate that only two targets are detected in coherent processing and Channel-A of binary processing. However, in the case of binary integration, the missed target is detected at Channel-B in Fig.~\ref{fig:range channels}c. On the contrary, the Range-OMP in Fig.~\ref{fig:range channels}d recovers all three targets at a higher resolution using only $y_{1,1,1}$. Hence, Range-OMP provides better estimates at lower complexity than DFT. Note that binary integration enhances performance only at low SNRs and exhibits the same performance as the coherent one at higher SNRs. Similarly, Range-OMP followed by binary integration can be used to detect ranges from all $\{y_{n,m,p}\}$ measurements, resulting in enhanced detection at low SNRs but with increased computational costs. Throughout our experiments, we consider the Range-OMP for $y_{1,1,1}$ only.
    \begin{table}
    \caption{Simulation Parameters}
    \label{tbl:parameters}
    \centering
    \begin{tabular}{p{4.5cm}p{2.7cm}}
    \hline\noalign{\smallskip}
    \textbf{Parameter} & \textbf{Value} \\
    \noalign{\smallskip}
    \hline
    \noalign{\smallskip}
    Carrier Frequency $f_{c}$ & 24 GHz \\
    Chirp Bandwidth $B=\gamma T_{c}$ & 250 MHz\\
    Chirp Duration $T_{c}$ & 40 $\mu$s\\
    Sampling frequency $f_{s}$ & 5 MHz\\
    Number of Chirps & $P=10$, $P_{max}=32$\\
    Aperture Length & $A_{T}=A_{R}=6\lambda$\\
    Number of antenna elements in ULA & $N_{T,f}=4$, $N_{R,f}=8$\\
    Number of antenna elements in SLA & $N_{T}=2$, $N_{R}=4$\\
    \noalign{\smallskip}
    \hline\noalign{\smallskip}
    \end{tabular}
    \end{table}
\begin{figure*}
\centering
\begin{subfigure}{0.32\textwidth}
    \includegraphics[width=\textwidth]{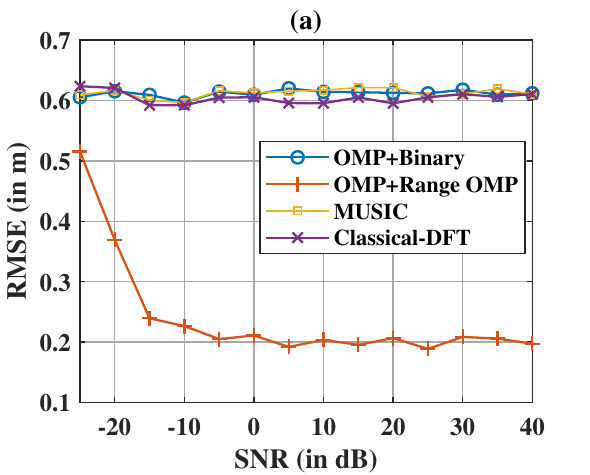}
\end{subfigure}
\hfill
\begin{subfigure}{0.32\textwidth}
    \includegraphics[width=\textwidth]{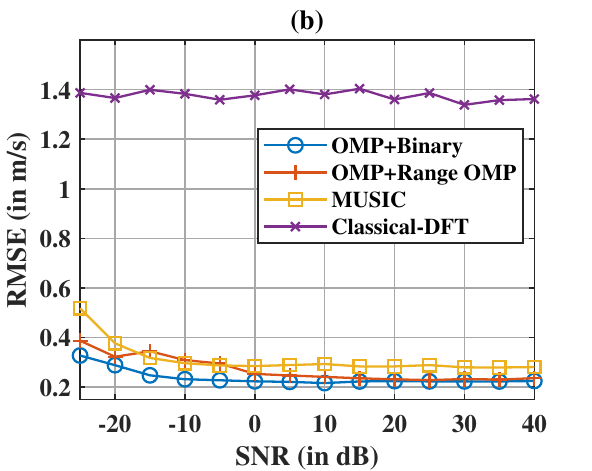}
\end{subfigure}
\hfill
\begin{subfigure}{0.32\textwidth}
    \includegraphics[width=\textwidth]{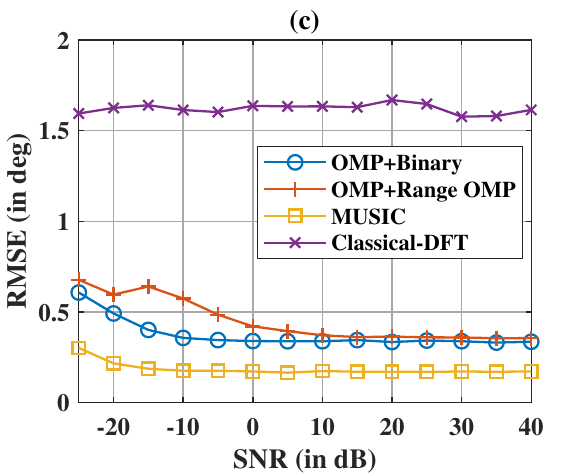}
\end{subfigure}
\caption{RMSE in (a) range, (b) velocity, and (c) AOA estimation at different SNRs for the OMP+Binary, OMP+Range-OMP, classical-DFT and MUSIC methods.}
\label{fig:rmse OMP DFT MUSIC}
\end{figure*}

\subsection{Range, velocity and AOA estimation}\label{subsec:overall numerical}
We consider $K=5$ targets with uniformly drawn ranges $R_{k}\in[20\; m, 120\; m]$, velocities $\nu_{k}\in[-78\; m/s,78\; m/s]$ and AOAs $\theta_{k}\in[-20^{\circ},20^{\circ}]$. We compare our proposed CS-based methods for SLA and sparse chirps with the conventional full-measurement radar systems in terms of detection probabilities, estimation errors, ROC, and run-time complexity. Table-\ref{tbl:methods} describes various methods considered for comparison. For classical-DFT, the range, velocity, and angular resolutions are $0.6$m, $4.88\;m/s$, and $7^{\circ}$, respectively, for the parameters given in Table~\ref{tbl:parameters}. In Range-OMP, a grid of resolution $0.12$m is chosen as in Section~\ref{subsec:range numerical}. Similarly, the velocity grid $\rho_{1\leq g\leq G_{D}}$ is constructed using $200$ uniform points in $[-78\; m/s, 78\; m/s]$ with resolution $0.78\; m/s$. The angular grid $\phi_{1\leq g\leq G_{\theta}}$ is a uniform grid (in $\sin\theta$ domain) of $50$ points in $[-30^{\circ},30^{\circ}]$ resulting in an approximate angular resolution of $2^{\circ}$. The CS-based AOA estimation method proposed in \cite{rossi2013spatial} (for pulsed-wave radars) assumes a known sparsity level and uses prior knowledge of the actual number of targets $K$ in the CS algorithms. On the contrary, in Range-OMP \eqref{eqn:range omp}, 1D-OMP \eqref{eqn:angle-Doppler 1D vectorized}, and 2D-OMP \eqref{eqn:angle-Doppler 2D} methods, we considered a sparsity level of $K_{max}=20$ to terminate the algorithms, while BP and LASSO considered $100$ iterations and an optimization tolerance of $10^{-4}$. The regularization parameter for the LASSO method was set to $0.6$. A threshold detection is then used to obtain the target parameters from the recovered signal. Note that, consistent with standard CS algorithms, the computational complexity of these methods is also determined by the chosen sparsity level $K_{max}$ and the number of iterations, as detailed in Table~\ref{tbl:complexity}. Additionally, while these parameters influence detection performance, their impact on estimation error is minimal. The estimation accuracy is instead primarily determined by the resolution of the corresponding grids.
    \begin{table}
    \caption{Various  parameter estimation methods}
    \label{tbl:methods}
    \centering
    \begin{tabular}{p{2.2cm}p{2.0cm}p{3.5cm}}
    \hline\noalign{\smallskip}
    \textbf{Method} & \textbf{Radar setup} & \textbf{Description}\\
    \noalign{\smallskip}
    \hline
    \noalign{\smallskip}
Classical-DFT & ULA with $P_{max}$ chirps & DFT in range, Doppler and angular domains\\
MUSIC & ULA with $P_{max}$ chirps & DFT in range and 2D-MUSIC with spatial smoothing\cite{belfiori20122d} in Doppler-angle domain\\
OMP+Binary & SLA with sparse $P$ chirps & DFT focusing with binary integration for range and 1D-OMP with \eqref{eqn:angle-Doppler 1D vectorized} for velocity \& AOA\\
OMP+Range-OMP & SLA with sparse $P$ chirps & Range-OMP and 1D-OMP with \eqref{eqn:angle-Doppler 1D vectorized} for velocity \& AOA\\
2D-OMP & SLA with sparse $P$ chirps & DFT-with-binary for range and 2D-OMP with \eqref{eqn:angle-Doppler 2D} for velocity \& AOA\\
BP & SLA with sparse $P$ chirps & DFT-with-binary for range and 1D-BP with \eqref{eqn:angle-Doppler 1D vectorized} for velocity \& AOA\\
LASSO & SLA with sparse $P$ chirps & DFT-with-binary for range and 1D-LASSO with \eqref{eqn:angle-Doppler 1D vectorized} for velocity \& AOA\\
    \noalign{\smallskip}
    \hline\noalign{\smallskip}
    \end{tabular}
    \end{table}
 
Since different methods have different range, Doppler and angular resolutions, we investigate the target detection and parameter estimation performances separately. To this end, we define a detected target as `hit' if the estimated range, velocity, and AOA from a given method are within the corresponding resolution from the true target parameters. Otherwise, it is a false alarm. We compare the hit rates of various methods to demonstrate target detection performance while different thresholds are set to maintain a near-constant false-alarm rate at all SNRs. For parameter estimation, we compute root-mean-squared estimation error (RMSE) for the targets classified as hits. Unless mentioned otherwise, the rates and RMSEs are averaged over $300$ independent runs.
\begin{figure}
\centering
\begin{subfigure}{0.35\textwidth}
    \includegraphics[width=\textwidth]{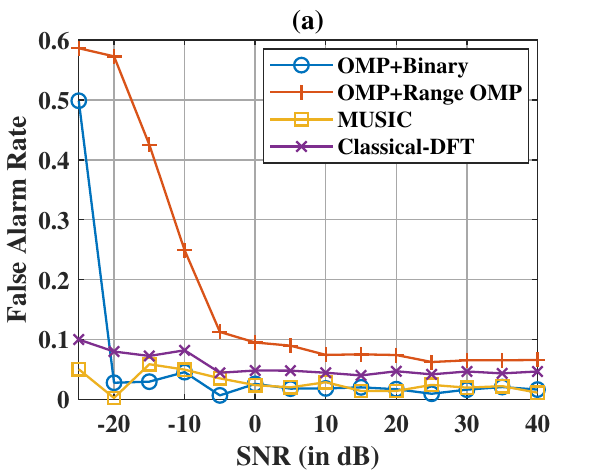}
\end{subfigure}
\hfill
\begin{subfigure}{0.35\textwidth}
    \includegraphics[width=\textwidth]{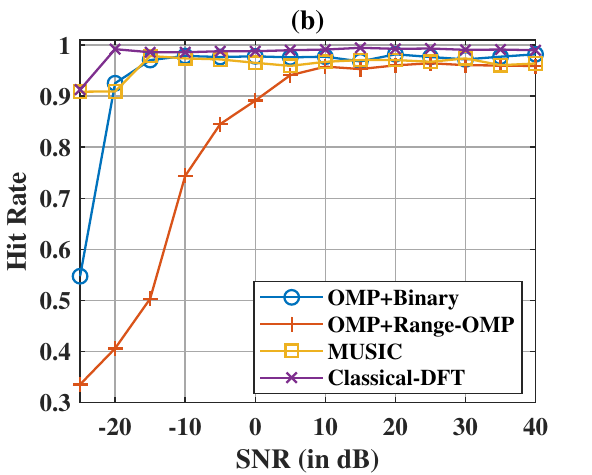}
\end{subfigure}
\caption{(a) False-alarm, and (b) hit rates at different SNRs for OMP+Binary, OMP+Range-OMP, classical-DFT, and MUSIC.}
\label{fig:rates OMP DFT MUSIC}
\end{figure}

\noindent\textbf{1) Comparison with classical-DFT \& MUSIC:} Fig.~\ref{fig:rates OMP DFT MUSIC} compares the false-alarm and hit rates of various methods over $-25$dB to $40$dB SNRs. The corresponding range, velocity, and angle RMSEs are shown in Fig.~\ref{fig:rmse OMP DFT MUSIC}. Fig.~\ref{fig:rates OMP DFT MUSIC} demonstrates that our proposed `OMP+Binary' method achieves similar false-alarm and hit rates as the classical-DFT and MUSIC for SNRs above $-20$ dB, but with only half the number of antenna elements and less than half transmitted chirps. Note that the performance degradation at very low SNRs, compared to uniform measurement systems, is a known limitation of radar frameworks relying on sparse signal processing \cite{mishra2019sub, xu2023automotive,na2018tendsur}. The ULA-based classical-DFT and MUSIC achieve a hit rate of $0.9$ even at $-25$ dB, benefiting from the larger number of measurements. However, practical MIMO-FMCW radars usually operate in the $0-20$ dB SNR regime at the receiver's input. In particular, OMP's sparse radar setup considers $6$ elements ($2$ transmitters and $4$ receivers), while the classical-DFT and MUSIC consider a ULA of $12$ elements ($4$ transmitters and $8$ receivers). Similarly, the conventional setup transmits $32$ chirps in a CPI, while in the sparse scenario, only $10$ are transmitted randomly. Range-OMP provides similar performance only at high SNRs, but at lower computational costs than DFT-focusing as indicated in Table~\ref{tbl:time}.

In Fig.~\ref{fig:rmse OMP DFT MUSIC}a, Range-OMP, owing to its fine-resolution grid, exhibits significantly lower range estimation error than the DFT approach employed in other methods. Similarly, in Fig.~\ref{fig:rmse OMP DFT MUSIC}b-c, the fine-resolution velocity and angle grids result in lower RMSEs than classical-DFT, but comparable to MUSIC. However, MUSIC, being a subspace-based approach, provides higher resolution with full measurements and higher computational costs. Note that while SNR primarily influences the hit and false alarm rates of a method, the estimation error is largely determined by the resolution associated with the method's parameter grids.\\
\begin{figure}
\centering
    \includegraphics[width=0.8\columnwidth]{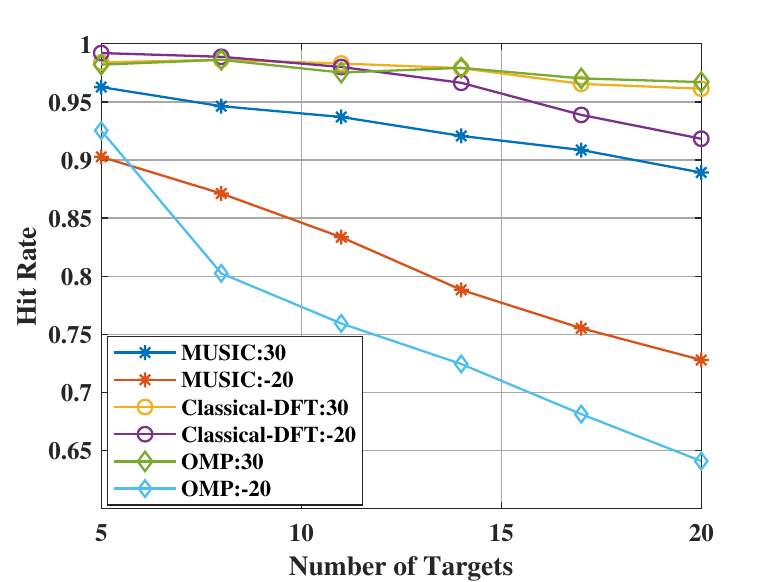}
\caption{Hit rate with varying number of targets $K$ at SNR= 30dB, \& -20dB, for OMP+Binary, classical-DFT and MUSIC.}
\label{fig:hit rate vs K}
\end{figure}
\begin{figure}
\centering
    \includegraphics[width=1.0\columnwidth]{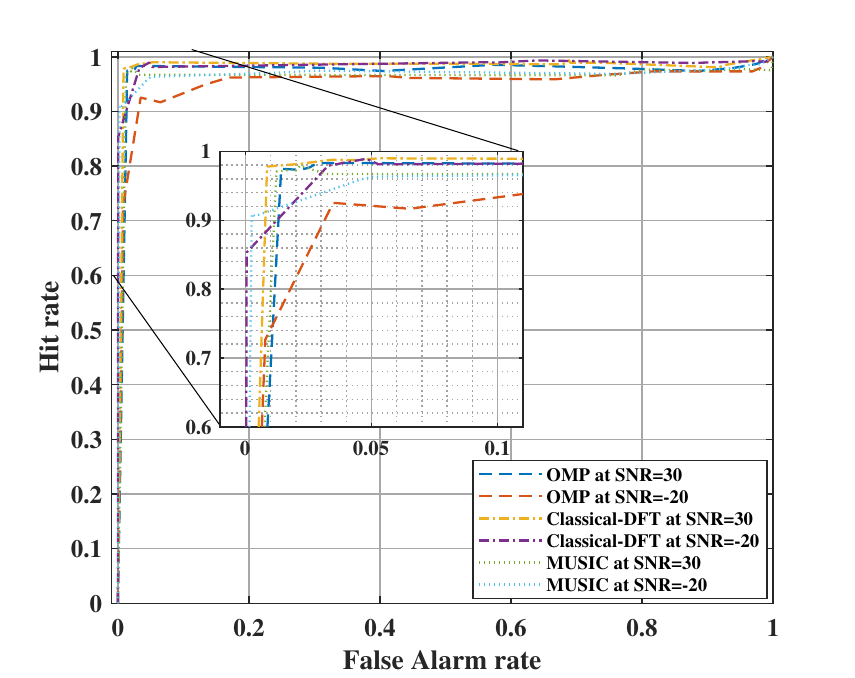}
\caption{ROC curves for OMP+Binary, classical-DFT, and MUSIC methods at $30$dB and $-20$dB SNRs.}
\label{fig:ROC}
\end{figure}
\textbf{2) Variation with number of targets:} Fig.~\ref{fig:hit rate vs K} shows the hit rate for different number of targets $K$ at $30$dB and $-20$dB SNRs, for OMP+Binary, classical-DFT and MUSIC methods. The false-alarm rate is kept constant at $0.05$. At high SNR, our proposed OMP method surpasses MUSIC with a hit rate greater than $0.95$, like classical DFT. As $K$ increases, MUSIC's hit rate decreases because of the fixed smoothing parameters and higher correlation among targets. On the other hand, OMP and classical-DFT maintain a near-constant hit rate. At very low SNR ($-20$dB), both OMP and MUSIC's hit rate degrades rapidly with $K$, unlike classical-DFT. While at $30$dB SNR, our proposed method surpasses MUSIC for all $K$, MUSIC exceeds at $-20$dB. OMP's recovery depends heavily on the measurement matrices and hence, the velocity and angular grids. A coarser grid with small correlation among dictionary elements can enhance the hit rate at low SNRs, trading off resolution.\\
\begin{figure}
\centering
\begin{subfigure}{0.35\textwidth}
    \includegraphics[width=\textwidth]{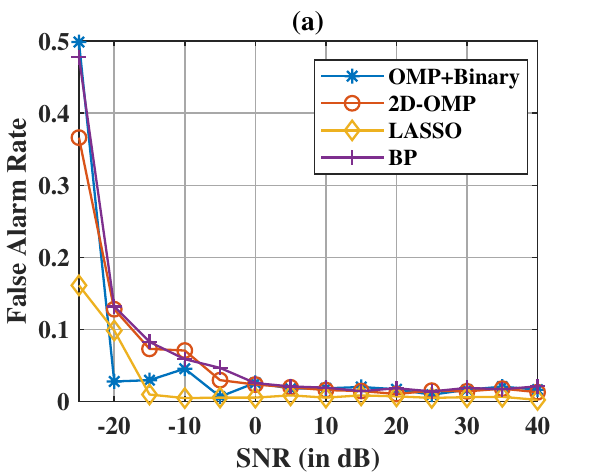}
\end{subfigure}
\hfill
\begin{subfigure}{0.35\textwidth}
    \includegraphics[width=\textwidth]{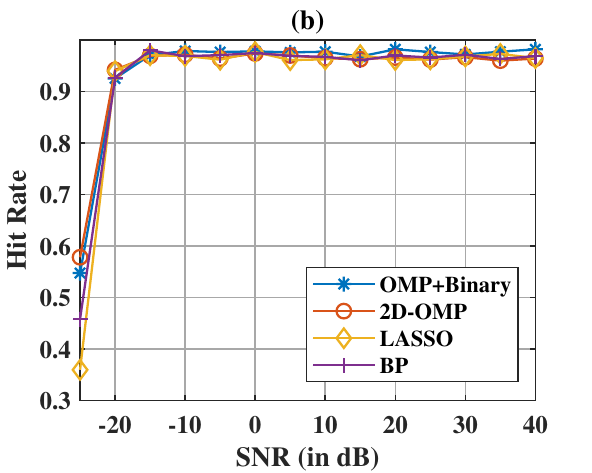}
\end{subfigure}
\caption{(a) False alarm and (b) hit rates at different SNRs for OMP+Binary, 2D-OMP, BP, and LASSO methods.}
\label{fig:rates CS methods}
\end{figure}
    \begin{table}
    \caption{Run-time for different methods (at $15$dB SNR)}
    \label{tbl:time}
    \centering
    \begin{tabular}{p{2.5cm}p{2.0cm}}
    \hline\noalign{\smallskip}
    \textbf{Method} & \textbf{Run time (in $s$)}\\
    \noalign{\smallskip}
    \hline
    \noalign{\smallskip}
    Classical-DFT & 0.2965 \\
    MUSIC & 0.6592 \\
    OMP+Binary & 0.2694 \\ 
    OMP+Range-OMP & 0.1582 \\
    2D-OMP & 0.1037 \\
    BP & 1.9015 \\
    LASSO & 0.1963 \\
    \noalign{\smallskip}
    \hline\noalign{\smallskip}
    \end{tabular}
    \end{table}
\textbf{3) ROCs:} Fig.~\ref{fig:ROC} shows the ROCs for OMP+Binary, classical-DFT and MUSIC methods at $30$dB and $-20$dB SNRs. All methods exhibit near-ideal ROCs at both low and high SNRs. Unlike classical-DFT, the detection performance of MUSIC and OMP degrades slightly at $-20$dB SNR, with OMP's ROC being below MUSIC. However, as mentioned earlier, MUSIC processes measurements from a $4\times 8$ ULA and all $32$ transmitted chirps, while our method considers only a $2\times 4$ SLA and $10$ chirps in a CPI. Also, from Table~\ref{tbl:time}, MUSIC has approximately $2.5$ times longer run-time than OMP.\\
\textbf{4) Comparison of different CS methods:} Fig.~\ref{fig:rates CS methods} illustrates the false-alarm and hit rates of 1-D OMP, BP, and LASSO as well as 2D-OMP methods. Because of the same velocity and angular grids, the RMSEs of these methods are the same and, hence, omitted here. All algorithms exhibit similar false-alarm and hit rates for high SNRs ($\geq 0$dB). On the other hand, at low SNRs ($-20$dB to $0$dB), LASSO achieves the same hit rate as other methods at the lowest false-alarm rate and hence, is more robust to noise. However, these methods differ in run-time complexity as elaborated in the following.\\
\textbf{5) Time complexity:} Table~\ref{tbl:time} provides the run-time (for one run)\footnote{We conducted our experiments on Dell Latitude 3410 with Intel Core i7 $10^{th}$ generation CPU @ 1.80 GHz clock speed and 8 GB RAM, using MATLAB (R2023a) on Windows 11.} for various methods listed in Table~\ref{tbl:methods}. As expected, MUSIC exhibits the longest run-time even though it provides high resolutions. The classical-DFT and OMP+Binary methods demonstrate similar time complexity, but OMP is more accurate, as indicated in Fig.~\ref{fig:rmse OMP DFT MUSIC}. Range-OMP further reduces the run-time of our method because it considers only measurement $y_{1,1,1}$ for range detection; see Section~\ref{subsec:range numerical}. Among different CS methods, 2D-OMP significantly outperforms the 1D OMP, BP, and LASSO algorithms. Overall, our CS-based methods, with 2D-OMP being the most efficient, achieve similar hit rates but higher resolutions than classical methods, even with limited measurements. Owing to the reduced number of measurements, our OMP methods are also faster than classical-DFT. Our Range-OMP further outperforms DFT-based range processing in terms of both resolution and computational costs.

\section{Summary}\label{sec:conclusions}
We have developed CS-based range, velocity, and AOA estimation algorithms for MIMO-FMCW radar with random SLA and sparsely transmitted chirps. We proposed two different range estimation techniques: (a) DFT-focusing with binary integration, which enhances detection at low SNRs, and (b) Range-OMP, which provides higher resolution and lower computational complexity. For joint Doppler-angle estimation, we considered both 1D-vectorized CS (OMP, BP, and LASSO) and 2D-OMP methods, wherein 2D-OMP directly estimates target parameters with lower complexity. Our numerical experiments show that these CS-based methods achieve the same performance as the classical-DFT and MUSIC approaches, but with fewer transmitter and receiver elements and less than half transmitted chirps. For joint Doppler-angle estimation, 2D-OMP provides higher resolutions than classical methods with superior efficiency, while Range-OMP surpasses DFT-based processing in both resolution and computational efficiency. Our theoretical analysis further demonstrated that the required number of chirps and antenna elements scales linearly with the number of targets and logarithmically with the grid points, offering high resolutions with fewer measurements. Unlike uniform recovery, the non-uniform recovery guarantee provides  flexibility in trading off between Doppler and angular measurements through a joint lower bound on chirps and antenna elements.

While the proposed framework demonstrates promising performance in simulation, several practical challenges must be addressed for real-world deployment. Hardware imperfections, such as oscillator phase noise, ADC quantization errors, RF non-linearities, and mutual coupling between antennas, can distort the idealized measurement model. Additionally, environmental factors, including multipath propagation, target occlusion, and clutter, introduce further complications. Large aperture arrays may also encounter spatial wideband and near-field effects, leading to range-angle coupling and non-planar wavefronts, which complicate the estimation process and contribute to model mismatch, thereby degrading sparse recovery performance. To address these issues, future research will focus on developing calibration-aware and model-adaptive CS techniques, robust dictionary designs that account for practical distortions, and estimation methods resilient to noise and hardware impairments. We plan to experimentally validate the proposed framework using the integrated radar evaluation module platforms to assess its effectiveness under realistic hardware and propagation conditions.

\appendices
\section{Proof of Lemma~\ref{lemma:toeplitz}}
\label{App-lemma-toeplitz}
Consider the $(i_{1},i_{2})$-th element of $\mathbf{Q}_{B}$ as $[\mathbf{Q}_{B}]_{i_{1},i_{2}}$. Recall from Section~\ref{sec:guarantees} (of the main paper) that any column of matrix $\mathbf{B}$ is given by $\mathbf{b}(\rho_{i})=[\exp{(j\frac{4\pi T_{c}}{\lambda}\rho_{i}\zeta_{1})},\hdots,\exp{(j\frac{4\pi T_{c}}{\lambda}\rho_{i}\zeta_{P})}]^{T}$ such that
\par\noindent\small
\begin{align*}
        [\mathbf{Q}_{B}]_{i_{1},i_{2}}=\mathbf{B}_{:,i_{1}}^{H}\mathbf{B}_{:,i_{2}}=\sum_{p=1}^{P}\exp{\left(j\frac{4\pi T_{c}}{\lambda}(\rho_{i_{2}}-\rho_{i_{1}})\zeta_{p}\right)}.
\end{align*}
\normalsize
If $\rho_{1:G_{D}}$ is a uniform grid, then $\rho_{i_{2}}-\rho_{i_{1}}$ is constant whenever $i_{2}-i_{1}$ is constant, i.e., along every diagonal of $\mathbf{Q}_{B}$. Hence, $\mathbf{Q}_{B}$ is a Toeplitz matrix. Also, using $i_{1}=i_{2}$, we obtain all the main diagonal elements of $\mathbf{Q}_{B}$ as $P$. Similarly, any column of matrix $\mathbf{C}$ is given by \small$\mathbf{c}(\phi_{i})=[\exp{(j \frac{\pi A}{\lambda}\sin{\phi_{i}}(\alpha_{1}+\beta_{1}))},\hdots,\exp{(j \frac{\pi A}{\lambda}\sin{\phi_{i}}(\alpha_{N_{T}}+\beta_{N_{R}}))}]^{T}$.\normalsize\\ Hence, under the uniform grid assumption, $\mathbf{Q}_{C}$ can also be proved to be a Toeplitz matrix with all main diagonal elements as $N_{T}N_{R}$ following similar arguments.

\section{Proof of Theorem~\ref{theorem:coherence}}
\label{App-thm-coherence}
In the following, Section~\ref{subsec:preliminary} provides some preliminary results, including suitable bounds on $\Gamma_{B}$ and $\Gamma_{C}$. The bound on coherence $\mu$ is then derived in Section~\ref{subsec:mu proof}.

\subsection{Preliminaries}\label{subsec:preliminary}
\begin{lemma}\label{lemma:ccdf B}
    Consider a radar transmitting $P$ chirps randomly over a total CPI comprising of $P_{max}$ chirps duration. Let the chirps' indices $\zeta_{1:P}\in \{0,1,\hdots,P_{max}-1\}$ be distributed as $\mathcal{P}_{p}$. which is symmetrical about $(P_{max}-1)/2$. If $\mathcal{P}_{p}$ and uniform grid $\rho_{1:G_{D}}$ satisfy \eqref{eqn:coherence v condition}, then
    \par\noindent\small
    \begin{align*}
    \mathbb{P}\left(\frac{1}{P}|\mathbf{B}_{:,1}^{H}\mathbf{B}_{:,i}|>\vartheta\right)=e^{-\vartheta^{2}P},
    \end{align*}
    \normalsize
    where $0<\vartheta<1$ and $i=2,3,\hdots,G_{D}$.
\end{lemma}
\begin{proof}
    Define $\zeta'_{p}=\zeta_{p}-\frac{P_{max}-1}{2}$. Substituting this in \eqref{eqn:gamma b}, we obtain $\Gamma_{B}(u)=\exp{(ju\frac{P_{max}-1}{2})}\Gamma'_{B}(u)$ where $\Gamma'_{B}(u)=\frac{1}{P}\sum_{p=1}^{P}\exp{(ju\zeta'_{p})}$. Now, since $\zeta_{p}\sim\mathcal{P}_{p}(\cdot)$, the random variable  $\zeta'_{p} $ has an even distribution because $\mathcal{P}_{p}(\cdot)$ is symmetric about $(P_{max}-1)/2$. Consequently, $\Gamma'_{B}(u)$ is asymptotically jointly Gaussian distributed \cite[Sec.~II]{lo1964mathematical} as
    \par\noindent\small
    \begin{align*}
    \begin{bmatrix}
            \textrm{Re}\;\Gamma'_{B}\\
            \textrm{Im}\;\Gamma'_{B}
        \end{bmatrix}\sim\mathcal{N}\left(\begin{bmatrix}
            \textrm{Re}\;\Psi_{p'}\\
            \textrm{Im}\;\Psi_{p'}
        \end{bmatrix},\begin{bmatrix}
            \sigma_{1}^{2}&0\\0&\sigma_{2}^{2}
        \end{bmatrix}\right),
    \end{align*}
    \normalsize
    where $\sigma_{1}^{2}(u)=\frac{1}{2P}(1+\Psi_{p'}(2u))-\frac{1}{P}\Psi_{p'}(u)^{2}$ and $\sigma_{2}^{2}(u)=\frac{1}{2P}(1-\Psi_{p'}(2u))$ with $\Psi_{p'}(\cdot)$ denoting the characteristic function given by $\Psi_{p'}(u)=\mathbb{E}_{\zeta\sim\mathcal{P}_{p}}[\exp{(ju(\zeta-\frac{P_{max}-1}{2}))}]=\exp{(-ju\frac{P_{max}-1}{2})}\Psi_{p}(u)$. Further using \eqref{eqn:coherence v condition}, we obtain $\Psi_{p'}(u^{D}_{1,i})=\Psi_{p'}(2u^{D}_{1,i})=0$ for $i=2,3,\hdots,G_{D}$. Hence, for $u=u^{D}_{1,i}$, $\Gamma'_{B}$ is Gaussian distributed as
    \par\noindent\small
    \begin{align*}
    \begin{bmatrix}
            \textrm{Re}\;\Gamma'_{B}\\
            \textrm{Im}\;\Gamma'_{B}
        \end{bmatrix}\sim\mathcal{N}\left(\begin{bmatrix}
            0\\
            0
        \end{bmatrix},\begin{bmatrix}
            1/2P&0\\0&1/2P
        \end{bmatrix}\right),
    \end{align*}
    \normalsize
    which implies $|\Gamma'_{B}|$ is Rayleigh distributed with $\sigma^{2}=1/2P$.

    Now, $|\Gamma_{B}(u)|=|\exp{(ju\frac{P_{max}-1}{2})}|\times |\Gamma'_{B}(u)|=|\Gamma'_{B}(u)|$. Therefore, for $u=u^{D}_{1,i}$, $|\Gamma_{B}|$ is also Rayleigh distributed with $\sigma^{2}=1/2P$. Finally, using \eqref{eqn:gamma b} and $\|\mathbf{B}_{:,i}\|_{2}=\sqrt{P}$, we have
    \par\noindent\small
    \begin{align*}
    \mathbb{P}\left(\frac{1}{P}|\mathbf{B}_{:,1}^{H}\mathbf{B}_{:,i}|>\vartheta\right)=\mathbb{P}(|\Gamma_{B}(u^{D}_{1,i})|>\vartheta)=e^{-\vartheta^{2}P},
    \end{align*}
    \normalsize
    where the last equality follows from the complementary cdf of a Rayleigh distribution.
\end{proof}
\begin{lemma}\label{lemma:ccdf C}
Consider a MIMO radar with random SLA as described in Section~\ref{sec:system model} (of the main paper) with $\alpha_{1:N_{T}}$ and $\beta_{1:N_{R}}$ drawn i.i.d. from even distributions $\mathcal{P}_{\alpha}$ and $\mathcal{P}_{\beta}$, respectively. If $\mathcal{P}_{\alpha}$, $\mathcal{P}_{\beta}$ and the uniform grid $\phi_{1:G_{\theta}}$ satisfy \eqref{eqn:coherence alpha beta condition}, then for $0<\vartheta<1$ and $i=2,3,\hdots,G_{\theta}$,\\
\textbf{1)} If $\alpha_{1:N_{T}}$ and $\beta_{1:N_{R}}$ are independent, we have:
\par\noindent\small
    \begin{align*}
    \mathbb{P}\left(\frac{1}{N_{T}N_{R}}|\mathbf{C}_{:,1}^{H}\mathbf{C}_{:,i}|>\vartheta\right)<2\vartheta\sqrt{N_{T}N_{R}}\mathcal{K}_{1}(2\vartheta\sqrt{N_{T}N_{R}}),
    \end{align*}
    \normalsize
    where $\mathcal{K}_{1}(\cdot)$ is the modified Bessel function of second kind.\\
\textbf{2)} If $N_{T}=N_{R}$ and $\alpha_{n}=\beta_{n}$ for all $1\leq n\leq N_{T}$,
\par\noindent\small
    \begin{align*}
    \mathbb{P}\left(\frac{1}{N_{T}^{2}}|\mathbf{C}_{:,1}^{H}\mathbf{C}_{:,i}|>\vartheta\right)=e^{-N_{T}\vartheta}.
    \end{align*}
    \normalsize
\end{lemma}
\begin{proof}
    The lemma is obtained trivially by generalizing the results of \cite[Theorem~1]{rossi2013spatial} to the MIMO-FMCW radar array.
\end{proof}

\subsection{Proof of the theorem}\label{subsec:mu proof}
As a consequence of Lemma~\ref{lemma:toeplitz} and substituting $\|\mathbf{B}_{:,i}\|_{2}=\sqrt{P}$ and $\|\mathbf{C}_{:,j}\|=\sqrt{N_{T}N_{R}}$ in \eqref{eqn:mu expansion}, we obtain
\par\noindent\small
\begin{align*}
    \mu\leq\textrm{max}\lbrace\textrm{max}_{i>1}\frac{1}{P}|\mathbf{B}_{:,1}^{H}\mathbf{B}_{:,i}|,\textrm{max}_{j>1}\frac{1}{N_{T}N_{R}}|\mathbf{C}_{:,1}^{H}\mathbf{C}_{:,j}|\rbrace.
\end{align*}
\normalsize
Hence, we can bound $\mathbb{P}(\mu\leq\vartheta)$ as
\par\noindent\small
\begin{align*}
    &\mathbb{P}(\mu\leq\vartheta)\\
    &\geq\mathbb{P}\left(\textrm{max}\lbrace\textrm{max}_{i>1}\frac{1}{P}|\mathbf{B}_{:,1}^{H}\mathbf{B}_{:,i}|,\textrm{max}_{j>1}\frac{1}{N_{T}N_{R}}|\mathbf{C}_{:,1}^{H}\mathbf{C}_{:,j}|\rbrace\leq\vartheta\right)\\
    &=\mathbb{P}\left(\textrm{max}_{i>1}\frac{1}{P}|\mathbf{B}_{:,1}^{H}\mathbf{B}_{:,i}|\leq\vartheta\right)\mathbb{P}\left(\textrm{max}_{j>1}\frac{1}{N_{T}N_{R}}|\mathbf{C}_{:,1}^{H}\mathbf{C}_{:,j}|\leq\vartheta\right)\\
    &=\mathbb{P}\left(\frac{1}{P}|\mathbf{B}_{:,1}^{H}\mathbf{B}_{:,i}|\leq\vartheta\;\forall\; i>1\right)\mathbb{P}\left(\frac{1}{N_{T}N_{R}}|\mathbf{C}_{:,1}^{H}\mathbf{C}_{:,j}|\leq\vartheta\;\forall\; j>1\right)\\
    &=\mathbb{P}\left(\frac{1}{P}|\mathbf{B}_{:,1}^{H}\mathbf{B}_{:,i}|\leq\vartheta\right)^{G_{D}-1}\mathbb{P}\left(\frac{1}{N_{T}N_{R}}|\mathbf{C}_{:,1}^{H}\mathbf{C}_{:,j}|\leq\vartheta\right)^{G_{\theta}-1},
\end{align*}
\normalsize
because $\mathbf{B}_{:,1}^{H}\mathbf{B}_{:,i}$ and $\mathbf{C}_{:,1}^{H}\mathbf{C}_{:,j}$ are mutually independent and identically distributed for all $i=2,3,\hdots,G_{D}$ and $j=2,3,\hdots,G_{\theta}$, respectively. Finally, using $\mathbb{P}(\mu>\vartheta)=1-\mathbb{P}(\mu\leq\vartheta)$ and Lemmas~\ref{lemma:ccdf B}-\ref{lemma:ccdf C}, we obtain \eqref{eqn:coherence bound independent case} and \eqref{eqn:coherence transceiver case}, respectively, when $\alpha_{1:N_{T}}$ and $\beta_{1:N_{R}}$ are independent and $\alpha_{n}=\beta_{n}\;\forall\;n$.

\section{Proof of Theorem~\ref{theorem:uniform}}
\label{App-thm-uniform}
According to \cite[Theorem~2.7]{rauhut2010compressive}, any $K$-sparse signal can be recovered if the measurement matrix $\mathbf{D}$ has RIP constant $\delta_{2K}<2/(3+\sqrt{7/4})\doteq\Lambda$. Hence, we need to bound $\delta_{2K}$ with a probability higher than $1-\epsilon$, or equivalently, $\mathbb{P}(\delta_{2K}>\Lambda)<\epsilon$ where $0<\epsilon<1$. However, RIP constant $\delta_{2K}$ and coherence $\mu$ of a matrix satisfy $\delta_{2K}\leq(2K-1)\mu$ \cite{rauhut2010compressive}. Hence, $\mathbb{P}(\delta_{2K}>\Lambda) \leq \mathbb{P}(\mu>\Lambda/(2K-1))$.

We first consider the case when $\alpha_{1:N_{T}}$ and $\beta_{1:N_{R}}$ are independent. Using \eqref{eqn:coherence bound independent case}, we obtain
\par\noindent\small
\begin{align}
\mathbb{P}(\delta_{2K}>\Lambda)<1-&(1-e^{-\vartheta_{0}^{2} P})^{G_{D}-1}\nonumber\\
&\times(1-2\vartheta_{0}\sqrt{N_{T}N_{R}}\;\mathcal{K}_{1}(2\vartheta_{0}\sqrt{N_{T}N_{R}}))^{G_{\theta}-1},\label{eqn:delta rhs}
\end{align}
\normalsize
where $\vartheta_{0}=\Lambda/(2K-1)$. Now, we need to find the number of chirps $P$ and antenna elements $N_{T}N_{R}$ such that the R.H.S. of \eqref{eqn:delta rhs} is less or equal to $\epsilon$. To this end, we approximate the Bessel function as $\mathcal{K}_{1}(q)\approx\sqrt{\frac{\pi}{2q}}\exp{(-q)}$ \cite{watson1922treatise} such that
\par\noindent\small
\begin{align*}
&1-(1-e^{-\vartheta_{0}^{2}P})^{G_{D}-1}(1-2\vartheta_{0}\sqrt{N_{T}N_{R}}\;\mathcal{K}_{1}(2\vartheta_{0}\sqrt{N_{T}N_{R}}))^{G_{\theta}-1}\\
&\approx 1-(1-e^{-\vartheta_{0}^{2}P})^{G_{D}-1}(1-(\pi^{2}\vartheta_{0}^{2}N_{T}N_{R})^{1/4}e^{-2\vartheta_{0}\sqrt{N_{T}N_{R}}})^{G_{\theta}-1}\\
&\approx 1-(1-G_{D}e^{-\vartheta_{0}^{2}P})(1-G_{\theta}(\pi^{2}\vartheta_{0}^{2}N_{T}N_{R})^{1/4}e^{-2\vartheta_{0}\sqrt{N_{T}N_{R}}}),
\end{align*}
\normalsize
where the last approximation is obtained using $(1-x)^{n}\approx 1-nx$ and $G_{\theta}, G_{D}\gg 1$. Hence, R.H.S. of \eqref{eqn:delta rhs} $\leq\epsilon$ implies
\par\noindent\small
\begin{align*}
(1-G_{D}e^{-\vartheta_{0}^{2}P})(1-G_{\theta}(\pi^{2}\vartheta_{0}^{2}N_{T}N_{R})^{1/4}e^{-2\vartheta_{0}\sqrt{N_{T}N_{R}}})\geq 1-\epsilon.
\end{align*}
\normalsize
Expanding L.H.S. with $G_{D}G_{\theta}(\pi^{2}\vartheta_{0}^{2}N_{T}N_{R})^{1/4}\exp{(-\vartheta_{0}^{2}P)}\\\times\exp{(-2\vartheta_{0}\sqrt{N_{T}N_{R}})}\approx 0$, we obtain
\par\noindent\small
\begin{align*}
G_{D}\exp{(-\vartheta_{0}^{2}P)}+G_{\theta}(\pi^{2}\vartheta_{0}^{2}N_{T}N_{R})^{1/4}\exp{(-2\vartheta_{0}\sqrt{N_{T}N_{R}})}\leq\epsilon,
\end{align*}
\normalsize
which is satisfied if $G_{D}\exp{(-\vartheta_{0}^{2}P)}\leq\epsilon_{1}$ and $G_{\theta}(\pi^{2}\vartheta_{0}^{2}N_{T}N_{R})^{1/4}\exp{(-2\vartheta_{0}\sqrt{N_{T}N_{R}})}\leq\epsilon_{2}$ with $\epsilon_{1}+\epsilon_{2}=\epsilon$.

Now, rearranging $G_{D}\exp{(-\vartheta_{0}^{2}P)}\leq\epsilon_{1}$ and substituting $\vartheta_{0}=\Lambda/(2K-1)$, we obtain $P\geq \frac{4}{\Lambda^{2}}\left(K-\frac{1}{2}\right)^{2}\log\left(\frac{G_{D}}{\epsilon_{1}}\right)$, which is \eqref{eqn:uniform P bound}. On the other hand, rearranging $G_{\theta}(\pi^{2}\vartheta_{0}^{2}N_{T}N_{R})^{1/4}\exp{(-2\vartheta_{0}\sqrt{N_{T}N_{R}})}\leq\epsilon_{2}$, we obtain
\par\noindent\small
\begin{align*}
-4\vartheta_{0}\sqrt{N_{T}N_{R}}\exp{(-4\vartheta_{0}\sqrt{N_{T}N_{R}})}\geq -\left(\frac{2\epsilon_{2}}{G_{\theta}\sqrt{\pi}}\right)^{2},
\end{align*}
\normalsize
which can be solved using Lambert W function\cite{corless1996lambert} as in the proof of \cite[Theorem~2]{rossi2013spatial} to obtain \eqref{eqn:uniform TR bound}. The constants $\kappa_{1}$ and $\kappa_{2}$ in \eqref{eqn:uniform P bound}-\eqref{eqn:uniform TR bound} are obtained by substituting $\Lambda=2/(3+\sqrt{7/4})$.

For the case when $\alpha_{n}=\beta_{n}$, we have $\mathbb{P}(\delta_{2K}>\Lambda)<1-(1-e^{-\vartheta_{0}^{2}P})^{G_{D}-1}(1-e^{-N_{T}\vartheta_{0}})^{G_{\theta}-1}$ using \eqref{eqn:coherence transceiver case} with $\vartheta_{0}=\Lambda/(2K-1)$. Again, using $(1-x)^{n}\approx 1-nx$ and $G_{D}, G_{\theta}\gg 1$ and following similar steps as in the independent $\alpha_{1:N_{T}}$ and $\beta_{1:N_{R}}$ case, we obtain $G_{D}e^{-\vartheta_{0}^{2}P}+G_{\theta}e^{-N_{T}\vartheta_{0}}\leq\epsilon$, which is satisfied if $G_{D}e^{-\vartheta_{0}^{2}P}\leq\epsilon_{1}$ and $G_{\theta}e^{-N_{T}\vartheta_{0}}\leq\epsilon_{2}$ with $\epsilon_{1}+\epsilon_{2}=\epsilon$. Rearranging these inequalities, we obtain \eqref{eqn:uniform P bound} and \eqref{eqn:uniform T bound}, respectively. Finally, the claim of the theorem then follows from \cite[Theorem~2.7]{rauhut2010compressive} considering the exact $K$-sparse signal case.

\section{Proof of Theorem~\ref{theorem:isotropy}}
\label{App-thm-isotropy}
\textbf{If part:} Using $\mathbf{D}_{N_{T}N_{R}(i-1)+j,:}=\mathbf{B}_{i,:}\otimes\mathbf{C}_{j,:}$, we obtain $\mathbb{E}[\mathbf{D}_{l,:}^{H}\mathbf{D}_{l,:}]=\mathbb{E}[\mathbf{B}_{i,:}^{H}\mathbf{B}_{i,:}]\otimes\mathbb{E}[\mathbf{C}_{j,:}^{H}\mathbf{C}_{j,:}]$ for some $1\leq i\leq P$ and $1\leq j\leq N_{T}N_{R}$ because $\zeta_{1:P}$ is independent of $\alpha_{1:N_{T}}$ and $\beta_{1:N_{R}}$. We first consider $[\mathbf{B}_{i,:}^{H}\mathbf{B}_{i,:}]$. Here, $\zeta_{1:P}$ are identically distributed such that $\mathbb{E}[\mathbf{B}_{i,:}^{H}\mathbf{B}_{i,:}]$ does not depend on row index $i$ and $\mathbb{E}[\mathbf{B}_{i,:}^{H}\mathbf{B}_{i,:}]=\frac{1}{P}\sum_{i=1}^{P}\mathbb{E}[\mathbf{B}_{i,:}^{H}\mathbf{B}_{i,:}]$. But, by simple comparison of matrix elements, we can show that $\sum_{i=1}^{P}\mathbb{E}[\mathbf{B}_{i,:}^{H}\mathbf{B}_{i,:}]=\mathbf{Q}_{B}$. Hence, $\mathbb{E}[\mathbf{B}_{i,:}^{H}\mathbf{B}_{i,:}]=\widetilde{\mathbf{Q}}_{B}$ where $\widetilde{\mathbf{Q}}_{B}=(1/P)\mathbb{E}[\mathbf{Q}_{B}]$. As a consequence of Lemma~\ref{lemma:toeplitz}, we only consider the first row of $\widetilde{\mathbf{Q}}_{B}$. Note that $(1,i)$-th element of $\mathbf{Q}_{B}$ is $\mathbf{B}_{:,1}^{H}\mathbf{B}_{:,i}=P\Gamma_{B}(u^{D}_{1,i})$ using \eqref{eqn:gamma b}. Hence, for $1\leq i\leq G_{D}$,
\par\noindent\small
\begin{align*}
[\widetilde{\mathbf{Q}}_{B}]_{1,i}&=\mathbb{E}[\Gamma_{B}(u^{D}_{1,i})]=\frac{1}{P}\sum_{p=1}^{P}\mathbb{E}[\exp{(ju^{D}_{1,i}\zeta_{p})}]\\
&=\frac{1}{P}\sum_{p=1}^{P}\mathbb{E}[\exp{(ju^{D}_{1,i}\zeta)}]=\Psi_{p}(u^{D}_{1,i})
\end{align*}
\normalsize
where the second-last equality follows because $\zeta_{p}$ are i.i.d. for all $1\leq p\leq P$. Finally, if $\Psi_{p}(u^{D}_{1,i})=0$ for $i=2,3,\hdots,G_{D}$ according to \eqref{eqn:condition for isotropy}, then $\mathbb{E}[\mathbf{B}_{i,:}^{H}\mathbf{B}_{i,:}]=\widetilde{\mathbf{Q}}_{B}=\mathbf{I}_{G_{D}}$. Note that $[\widetilde{\mathbf{Q}}_{B}]_{1,1}=1$ because all main diagonal elements of $\mathbf{Q}_{B}$ is $P$ from Lemma~\ref{lemma:toeplitz}. Further, following similar steps as in proof of \cite[Theorem~3]{rossi2013spatial}, we can trivially show that $\mathbb{E}[\mathbf{C}_{j,:}^{H}\mathbf{C}_{j,:}]=\mathbf{I}_{G_{\theta}}$ if $\Psi_{\xi}(u^{\theta}_{1,j})=0$ for $j=2,3,\hdots,G_{\theta}$ as in \eqref{eqn:condition for isotropy}. Finally, $\mathbb{E}[\mathbf{D}_{l,:}^{H}\mathbf{D}_{l,:}]=\mathbb{E}[\mathbf{B}_{i,:}^{H}\mathbf{B}_{i,:}]\otimes\mathbb{E}[\mathbf{C}_{j,:}^{H}\mathbf{C}_{j,:}]=\mathbf{I}_{G_{D}}\otimes\mathbf{I}_{G_{\theta}}=\mathbf{I}_{G_{D}G_{\theta}}$, i.e., if \eqref{eqn:condition for isotropy} holds, matrix $\mathbf{D}$ satisfies isotropy property.

\textbf{Only if part:} If $\Psi_{p}(u^{D}_{1,i})\neq 0$ for some $i=2,3,\hdots,G_{D}$, then $[\widetilde{\mathbf{Q}}_{B}]_{1,i}\neq 0$. Hence, $\mathbb{E}[\mathbf{B}_{i,:}^{H}\mathbf{B}_{i,:}]\neq \mathbf{I}$ which implies $\mathbb{E}[\mathbf{D}_{l,:}^{H}\mathbf{D}_{l,:}]\neq\mathbf{I}$, i.e., matrix $\mathbf{D}$ does not satisfy isotropy property. Similarly, matrix $\mathbf{D}$ also does not satisfy isotropy property if $\Psi_{\xi}(u^{\theta}_{1,j})\neq 0$ for some $j=2,3,\hdots,G_{\theta}$. Therefore, \eqref{eqn:condition for isotropy} is a necessary condition.

\section{Proof of Proposition~\ref{prop:example}}
\label{App-prop-example}
Note that $\alpha_{1:N_{T}},\beta_{1:N_{R}}\in[-1/2,1/2]$ when $A_{T}=A_{R}=A/2$ and the characteristic function of uniform distributions $\mathcal{P}_{\alpha}$ and $\mathcal{P}_{\beta}$ are $\Psi_{\alpha}(u)=\Psi_{\beta}(u)=\frac{\sin(u/2)}{u/2}$. From \eqref{eqn:coherence alpha beta condition}, we need $\Psi_{\alpha}(u^{\theta}_{i,j})=\Psi_{\alpha}(2u^{\theta}_{i,j})=0$ and $\Psi_{\beta}(u^{\theta}_{i,j})=\Psi_{\beta}(2u^{\theta}_{i,j})=0$ for $i\neq j$ which can be satisfied if $u^{\theta}_{i,j}/2=k\pi$ for some integer $k$. But, by definition, $\frac{u^{\theta}_{i,j}}{2}=\pi\frac{A}{2\lambda}(\sin{\phi_{j}}-\sin{\phi_{i}})$. Hence, condition \eqref{eqn:coherence alpha beta condition} is satisfied if $\phi_{1:G_{\theta}}$ is a uniformly spaced grid of spacing $2\lambda/A$ in the $\sin{\theta}$ domain.
    
For discrete uniform distribution $\mathcal{P}_{p}$, we have
    \par\noindent\small
    \begin{align*}
        \Psi_{p}(u)=\frac{1-e^{jP_{max}u}}{P_{max}(1-e^{u})}.
    \end{align*}
    \normalsize
    We need $\Psi_{p}(u^{D}_{i,l})=\Psi_{p}(2u^{D}_{i,l})=0$ for $i\neq l$ which is satisfied if $P_{max}u^{D}_{i,l}=2k\pi$ for some integer $k$. But, $P_{max}u^{D}_{i,l}=2\pi\cdot\frac{2P_{max}T_{c}}{\lambda}(\rho_{l}-\rho_{i})$. Hence, \eqref{eqn:coherence v condition} is satisfied when $\rho_{1:G_{D}}$ is a uniform grid of spacing $\lambda/(2P_{max}T_{c})$.

\section{Effect of calibration errors on random SLA}
\label{App-calibration}
Throughout Section~\ref{sec:numericals}, both random SLA and ULA were assumed to be perfectly calibrated. In practice, however, calibration errors are unavoidable. Here, we examine the effect of imperfect transmitter and receiver calibration in random SLAs on the performance of the proposed methods. Specifically, we compare the sensitivity of OMP+Binary and OMP+Range-OMP against that of classical-DFT and MUSIC, evaluating their robustness to calibration errors in terms of detection rates and estimation accuracy. Recall from Table~\ref{tbl:methods} that both classical-DFT and MUSIC employ ULA, and target ranges are obtained via conventional DFT after coherent integration of all chirps and array channels. Velocities and AOAs are then estimated jointly using 2D-DFT in classical-DFT and 2D-MUSIC in the MUSIC method. In contrast, the proposed OMP-based methods rely on measurements from a random SLA. Target ranges are estimated via DFT-focusing with binary integration in OMP+Binary, and via Range-OMP in OMP+Range-OMP, while joint Doppler-angle estimation in both cases is performed using vectorized 1D-OMP.

\begin{figure}
\centering
\begin{subfigure}{0.35\textwidth}
    \includegraphics[width=\textwidth]{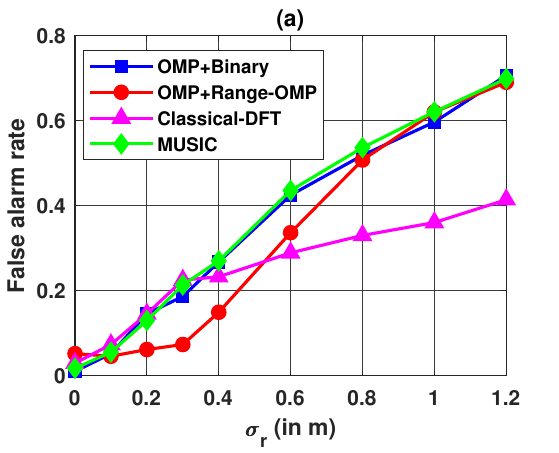}
\end{subfigure}
\hfill
\begin{subfigure}{0.35\textwidth}
    \includegraphics[width=\textwidth]{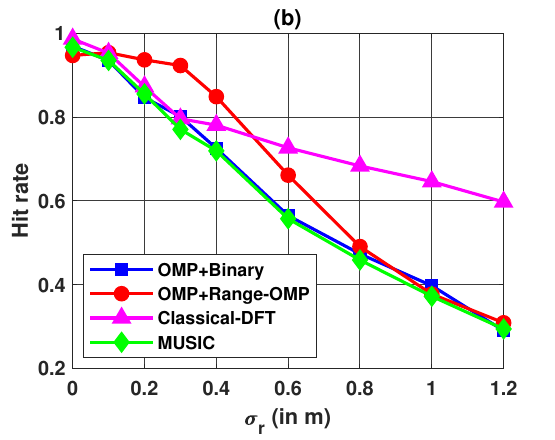}
\end{subfigure}
\caption{(a) False alarm, and (b) hit rates under varying $\sigma_{r}$ with $|\sigma_{\theta}|=2.5|\sigma_{r}|$ for OMP+Binary, OMP+Range-OMP, classical-DFT, and MUSIC methods.}
\label{fig:rates at cal err OMP DFT MUSIC}
\end{figure}
\begin{figure}
\centering
\begin{subfigure}{0.35\textwidth}
    \includegraphics[width=\textwidth]{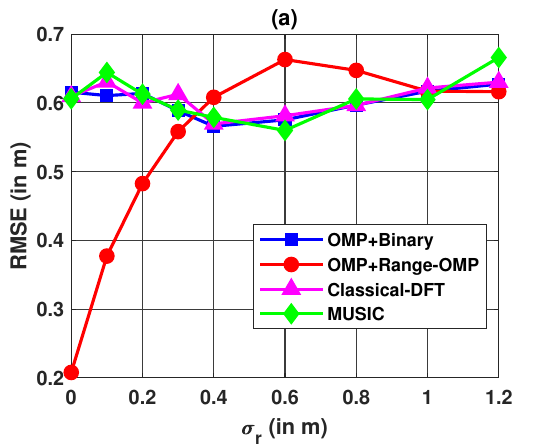}
\end{subfigure}
\hfill
\begin{subfigure}{0.35\textwidth}
    \includegraphics[width=\textwidth]{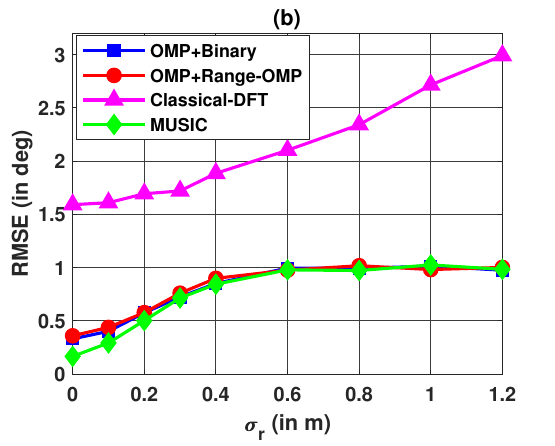}
\end{subfigure}
\caption{RMSE in (a) range and (b) AOA estimation under varying $\sigma_{r}$ with $|\sigma_{\theta}|=2.5|\sigma_{r}|$ for OMP+Binary, OMP+Range-OMP, classical-DFT, and MUSIC methods.}
\label{fig:rmse at cal err OMP DFT MUSIC}
\end{figure}

In order to simulate calibration imperfections, we introduce additional phase and delay errors into the received signal relative to a reference antenna element. For each target, i.i.d. phase errors $\Delta\theta_{n,m}$ are added to the received signal in \eqref{eqn:discrete IF signal}, with $\Delta\theta_{n,m}\sim\mathcal{N}(0,\sigma_{\theta}^{2})$. Similarly, i.i.d. delay errors $\Delta\tau^{k}$ are also added to $\tau_{n,m,p}^{k}$ in \eqref{eqn:discrete IF signal}, where $\Delta\tau^{k}$ is drawn from a zero-mean uniform distribution of variance $\sigma_{r}^{2}$, i.e., $\Delta\tau^{k}\sim\mathcal{U}[-\sqrt{3}\sigma_{r},\sqrt{3}\sigma_r]$. We vary $\sigma_{\theta}$ from $0.25^{\circ}$ to $3.0^{\circ}$, and $\sigma_{r}$ from $0.1$ m to $1.2$ m. We consider an SNR of $30$ dB and average the detection rates and RMSEs over $300$ independent runs. The false-alarm and hit rates of various methods under varying $\sigma_{\theta}$ and $\sigma_{r}$ are presented in Tables~\ref{tbl:false alarm} and \ref{tbl:hit alarm}, while the corresponding RMSEs in range and AOA estimation are presented in Tables~\ref{tbl:range rmse} and \ref{tbl:angle rmse}, respectively. For clarity of graphical illustration, we show the detection rates and estimation errors only for $|\sigma_{\theta}|=2.5|\sigma_{r}|$ where $\sigma_{r}\in[0,1.2]$ m. The results are illustrated in Fig.~\ref{fig:rates at cal err OMP DFT MUSIC} for false-alarm and hit rates, and in Fig.~\ref{fig:rmse at cal err OMP DFT MUSIC} for RMSEs under varying $\sigma_{r}$. Note that velocity estimation remains unaffected by these calibration errors, since no inter-chirp delay variations are considered within a CPI.
    \begin{table}
    \caption{False alarm rates under varying calibration errors for various methods}
    \label{tbl:false alarm}
    \centering
    \begin{tabular}{p{0.6cm}p{0.8cm}p{1.1cm}p{1.1cm}p{1.1cm}p{1.1cm}}
    \hline\noalign{\smallskip}
    \textbf{$\sigma_{\theta}$} & \textbf{$\sigma_{r}$} &\textbf{OMP +Binary} & \textbf{OMP +Range-OMP} & \textbf{Classical -DFT} & \textbf{MUSIC}\\
    \noalign{\smallskip}
    \hline
    \noalign{\smallskip}
    $0^{\circ}$ & $0$ m & \textbf{0.0108} & 0.0461 & 0.0297 & 0.0177\\
    $0.25^{\circ}$ & $0.1$ m & \textbf{0.0519} & 0.0523 & 0.0736 & 0.0556\\
    $0.25^{\circ}$ & $0.4$ m & 0.2348 & \textbf{0.0831} & 0.2400 & 0.2458\\
    $0.25^{\circ}$ & $0.8$ m & 0.3132 & \textbf{0.2837} & 0.3275 & 0.3301\\
    $0.25^{\circ}$ & $1.2$ m & \textbf{0.3766} & 0.4073 & 0.3837 & 0.3880\\
    $1^{\circ}$ & $0.1$ m & 0.1157 & 0.1091 & \textbf{0.0741} & 0.1064\\
    $1^{\circ}$ & $0.4$ m & 0.2678 & \textbf{0.1496} & 0.2331 & 0.2705\\
    $1^{\circ}$ & $0.8$ m & 0.3682 & \textbf{0.3254} & 0.3288 & 0.3564\\
    $1^{\circ}$ & $1.2$ m & 0.4138 & 0.4378 & \textbf{0.3862} & 0.4033\\
    $2^{\circ}$ & $0.1$ m & 0.3708 & 0.3548 & \textbf{0.0751} & 0.3861\\
    $2^{\circ}$ & $0.4$ m & 0.4878 & 0.3921 & \textbf{0.2596} & 0.4700\\
    $2^{\circ}$ & $0.8$ m & 0.5189 & 0.5078 & \textbf{0.3304} & 0.5369\\
    $2^{\circ}$ & $1.2$ m & 0.5826 & 0.5926 & \textbf{0.3698} & 0.5700\\
    $3^{\circ}$ & $0.1$ m & 0.5517 & 0.5078 & \textbf{0.1071} & 0.5346\\
    $3^{\circ}$ & $0.4$ m & 0.6481 & 0.5452 & \textbf{0.2457} & 0.6149\\
    $3^{\circ}$ & $0.8$ m & 0.6553 & 0.6278 & \textbf{0.3755} & 0.6521\\
    $3^{\circ}$ & $1.2$ m & 0.7067 & 0.6908 & \textbf{0.4144} & 0.6992\\
    \noalign{\smallskip}
    \hline\noalign{\smallskip}
    \end{tabular}
    \end{table}
    \begin{table}
    \caption{Hit rates under varying calibration errors for various methods}
    \label{tbl:hit alarm}
    \centering
    \begin{tabular}{p{0.6cm}p{0.8cm}p{1.1cm}p{1.1cm}p{1.1cm}p{1.1cm}}
    \hline\noalign{\smallskip}
    \textbf{$\sigma_{\theta}$} & \textbf{$\sigma_{r}$} &\textbf{OMP +Binary} & \textbf{OMP +Range-OMP} & \textbf{Classical -DFT} & \textbf{MUSIC}\\
    \noalign{\smallskip}
    \hline
    \noalign{\smallskip}
    $0^{\circ}$ & $0$ m & 0.9693 & 0.9537 & \textbf{0.9867} & 0.9667\\
    $0.25^{\circ}$ & $0.1$ m & 0.9353 & 0.9472 & \textbf{0.9527} & 0.9353\\
    $0.25^{\circ}$ & $0.4$ m & 0.7500 & \textbf{0.9140} & 0.7733 & 0.7440\\
    $0.25^{\circ}$ & $0.8$ m & 0.6773 & \textbf{0.7153} & 0.6867 & 0.6560\\
    $0.25^{\circ}$ & $1.2$ m & 0.6160 & 0.5880 & \textbf{0.6267} & 0.6073\\
    $1^{\circ}$ & $0.1$ m & 0.8667 & 0.8880 & \textbf{0.9393} & 0.8873\\
    $1^{\circ}$ & $0.4$ m & 0.7253 & \textbf{0.8488} & 0.7807 & 0.7193\\
    $1^{\circ}$ & $0.8$ m & 0.6240 & 0.6727 & \textbf{0.6847} & 0.6333\\
    $1^{\circ}$ & $1.2$ m & 0.5767 & 0.5585 & \textbf{0.6253} & 0.5887\\
    $2^{\circ}$ & $0.1$ m & 0.6233 & 0.6447 & \textbf{0.9467} & 0.6073\\
    $2^{\circ}$ & $0.4$ m & 0.5013 & 0.6060 & \textbf{0.7533} & 0.5180\\
    $2^{\circ}$ & $0.8$ m & 0.4733 & 0.4907 & \textbf{0.6833} & 0.4587\\
    $2^{\circ}$ & $1.2$ m & 0.4160 & 0.4053 & \textbf{0.6400} & 0.4233\\
    $3^{\circ}$ & $0.1$ m & 0.4440 & 0.4900 & \textbf{0.9107} & 0.4533\\
    $3^{\circ}$ & $0.4$ m & 0.3493 & 0.4540 & \textbf{0.7620} & 0.3813\\
    $3^{\circ}$ & $0.8$ m & 0.3367 & 0.3720 & \textbf{0.6340} & 0.3400\\
    $3^{\circ}$ & $1.2$ m & 0.2900 & 0.3087 & \textbf{0.5973} & 0.2940\\
    \noalign{\smallskip}
    \hline\noalign{\smallskip}
    \end{tabular}
    \end{table}

Recall from Fig.~\ref{fig:rates OMP DFT MUSIC} that at high SNR ($30$ dB), all methods achieve comparable detection performance. Fig.~\ref{fig:rates at cal err OMP DFT MUSIC} as well as Tables~\ref{tbl:false alarm}-\ref{tbl:hit alarm} demonstrate that as the phase and delay calibration errors increase, the detection capability of all methods deteriorates, i.e., hit rates decrease while false-alarm rates increase. However, in Fig.~\ref{fig:rates at cal err OMP DFT MUSIC}, while the proposed OMP+Binary method and MUSIC exhibit similar degradation with increasing $\sigma_{\theta}$ and $\sigma_{r}$, the classical-DFT method consistently achieves superior robustness, maintaining higher hit rates and significantly lower false-alarm rates even under severe calibration errors. The superior detection performance of the classical-DFT is also evident from Tables~\ref{tbl:false alarm}-\ref{tbl:hit alarm} (see highlighted false alarm and hit rates), where it achieves lower false alarm rates and higher hit rates than other methods for most $(\sigma_{\theta},\sigma_{r})$ pairs, except the low calibration error cases. Despite both classical-DFT and MUSIC using measurements from the same ULA, MUSIC is observed to be more sensitive to calibration errors, which is expected given its reliance on subspace-based estimation because calibration errors lead to loss of orthogonality. Interestingly, for low calibration error levels (e.g., $\sigma_{\theta}\leq 0.75^{\circ}$ and $\sigma_{r}\leq 0.3$ m in Fig.~\ref{fig:rates at cal err OMP DFT MUSIC}), OMP+Range-OMP preserves its detection performance and even outperforms OMP+Binary and MUSIC in terms of hit and false-alarm rates for $\sigma_{\theta}\leq 2^{\circ}$ and $\sigma_{r}\leq 0.8$ m. Similarly, OMP+Range-OMP shows improved detection rates than other methods in Tables~\ref{tbl:false alarm}-\ref{tbl:hit alarm} for certain low calibration error cases. This advantage can be attributed to the greedy OMP algorithm employed for range estimation, which is generally more resilient to mild dictionary mismatch and calibration errors.
    \begin{table}
    \caption{RMSEs (in m) in range estimation under varying calibration errors for various methods}
    \label{tbl:range rmse}
    \centering
    \begin{tabular}{p{0.6cm}p{0.8cm}p{1.1cm}p{1.1cm}p{1.1cm}p{1.1cm}}
    \hline\noalign{\smallskip}
    \textbf{$\sigma_{\theta}$} & \textbf{$\sigma_{r}$} &\textbf{OMP +Binary} & \textbf{OMP +Range-OMP} & \textbf{Classical -DFT} & \textbf{MUSIC}\\
    \noalign{\smallskip}
    \hline
    \noalign{\smallskip}
    $0^{\circ}$ & $0$ m & 0.5660 & \textbf{0.2077} & 0.5695 & 0.5600\\
    $0.25^{\circ}$ & $0.1$ m & 0.5754 & \textbf{0.3771} & 0.5814 &  0.5790\\
    $0.25^{\circ}$ & $0.4$ m & 0.5875 & 0.5916 & 0.6007 &  0.5886\\
    $0.25^{\circ}$ & $0.8$ m & 0.6232 & 0.6572 & 0.6149 &  0.6148\\
    $0.25^{\circ}$ & $1.2$ m & 0.6556 & 0.6777 & 0.6610 &  0.6466\\
    $1^{\circ}$ & $0.1$ m & 0.6243 & \textbf{0.3631} & 0.6379 &  0.6446\\
    $1^{\circ}$ & $0.4$ m & 0.6107 & \textbf{0.5583} & 0.6096 & 0.6059\\
    $1^{\circ}$ & $0.8$ m & 0.6021 & 0.6646 & 0.6007 & 0.6194\\
    $1^{\circ}$ & $1.2$ m & 0.6581 & 0.6792 & 0.6439 & 0.6502\\
    $2^{\circ}$ & $0.1$ m & 0.6277 & \textbf{0.3674} & 0.6470 &  0.6284\\
    $2^{\circ}$ & $0.4$ m & 0.5546 & 0.5837 & 0.5862 &  0.5838\\
    $2^{\circ}$ & $0.8$ m & 0.6157 & 0.6164 & 0.6218 &  0.6128\\
    $2^{\circ}$ & $1.2$ m & 0.6050 & 0.6492 & 0.6290 & 0.6582\\
    $3^{\circ}$ & $0.1$ m & 0.5887 & \textbf{0.3505} & 0.6296 &  0.6134\\
    $3^{\circ}$ & $0.4$ m & 0.5550 & 0.5933 & 0.5801 & 0.5299\\
    $3^{\circ}$ & $0.8$ m & 0.5747 & 0.6371 & 0.6013 & 0.5651\\
    $3^{\circ}$ & $1.2$ m & 0.6274 & 0.6632 & 0.6308 &  0.6661\\
    \noalign{\smallskip}
    \hline\noalign{\smallskip}
    \end{tabular}
    \end{table}
    \begin{table}
    \caption{RMSEs (in deg) in AOA estimation under varying calibration errors for various methods}
    \label{tbl:angle rmse}
    \centering
    \begin{tabular}{p{0.6cm}p{0.8cm}p{1.1cm}p{1.1cm}p{1.1cm}p{1.1cm}}
    \hline\noalign{\smallskip}
    \textbf{$\sigma_{\theta}$} & \textbf{$\sigma_{r}$} &\textbf{OMP +Binary} & \textbf{OMP +Range-OMP} & \textbf{Classical -DFT} & \textbf{MUSIC}\\
    \noalign{\smallskip}
    \hline
    \noalign{\smallskip}
    $0^{\circ}$ & $0$ m & 0.3304 & 0.3593 & 1.5960 & \textbf{0.1658}\\
    $0.25^{\circ}$ & $0.1$ m & 0.4015 & 0.4382 & 1.6093 & \textbf{0.2921}\\
    $0.25^{\circ}$ & $0.4$ m & 0.3980 & 0.4174 & 1.5651 & \textbf{0.2820}\\
    $0.25^{\circ}$ & $0.8$ m & 0.4012 & 0.4119 & 1.5783 & \textbf{0.2756}\\
    $0.25^{\circ}$ & $1.2$ m & 0.3759 & 0.4034 & 1.5926 & \textbf{0.2786}\\
    $1^{\circ}$ & $0.1$ m & 0.8598 & 0.8853 & 1.8634 & 0.8457\\
    $1^{\circ}$ & $0.4$ m & 0.8509 & 0.8989 & 1.8853 & 0.8438\\
    $1^{\circ}$ & $0.8$ m & 0.8723 & 0.8622 & 1.8140 & 0.8223\\
    $1^{\circ}$ & $1.2$ m & 0.8530 & 0.8078 & 1.8536 & 0.7967\\
    $2^{\circ}$ & $0.1$ m & 1.0327 & 1.0155 & 2.4414 & 0.9982\\
    $2^{\circ}$ & $0.4$ m & 0.9751 & 1.0111 & 2.4346 & 1.0129\\
    $2^{\circ}$ & $0.8$ m & 0.9928 & 0.9815 & 2.3415 & 0.9787\\
    $2^{\circ}$ & $1.2$ m & 1.0396 & 1.0134 & 2.3252 & 1.0565\\
    $3^{\circ}$ & $0.1$ m & 1.0365 & 1.0622 & 2.8653 & 1.0239\\
    $3^{\circ}$ & $0.4$ m & 1.0485 & 1.0526 & 2.8231 & 0.9877\\
    $3^{\circ}$ & $0.8$ m & 0.9760 & 1.0035 & 2.8243 & 1.0476\\
    $3^{\circ}$ & $1.2$ m & 0.9755 & 1.0173 & 2.9911 & 1.0247\\
    \noalign{\smallskip}
    \hline\noalign{\smallskip}
    \end{tabular}
    \end{table}

In the case of estimation accuracy, Fig.~\ref{fig:rmse at cal err OMP DFT MUSIC} shows that all methods experience similar increases in range and AOA errors as $\sigma_{\theta}$ and $\sigma_{r}$ increases. Nevertheless, Range-OMP achieves significantly lower range RMSE at small calibration errors owing to its finer range grid, while classical-DFT exhibits the highest angular RMSE due to its coarse angular resolution (see Section~\ref{subsec:overall numerical}.1). Table~\ref{tbl:range rmse} shows that OMP+Range-OMP achieves the lowest range errors for many small $\sigma_{r}$ values (see highlighted RMSEs), whereas other methods yield similar errors due to the range resolution being limited by the DFT-defined bins. Notably, Table~\ref{tbl:angle rmse} indicates that MUSIC provides significantly lower AOA errors for $\sigma_{\theta}\leq 0.25^{\circ}$ (see highlighted RMSEs), which aligns with the results in Fig.~\ref{fig:rmse OMP DFT MUSIC} of the main paper. Overall, these results suggest that the proposed random-SLA-based methods exhibit more sensitivity to joint phase and delay calibration errors than classical-DFT with ULA, but comparable to MUSIC. Importantly, Range-OMP demonstrates an ability to mitigate small calibration errors effectively through its robust CS-based processing.

\bibliographystyle{IEEEtran}
\bibliography{references}
\end{document}